\documentclass[conf]{IEEEtran}
\usepackage{amsthm,amsfonts,amssymb,mathtools,nccmath,mathrsfs}
\usepackage[lined,linesnumbered,ruled,commentsnumbered]{algorithm2e} 
\usepackage{pgfplots}
\usepackage{xcolor}
\usepackage{url}
\usepackage{array}
\usepackage{subfigure}
\usepackage{graphicx}
\usepackage{cite}
\usepackage{hyperref}
\usepackage[T1]{fontenc}
\usepackage[utf8]{inputenc}
\usepackage{nicefrac}
\usepackage{comment}
\usepackage{makecell}
\usepackage{textcomp}
\usepackage{stfloats}
\usepackage{setspace}
\usepackage{soul}
\usepackage{tikz}
\usepackage{relsize}

\allowdisplaybreaks[4] 

\usetikzlibrary{arrows,shapes,backgrounds,plotmarks,positioning}
\usetikzlibrary{decorations.pathreplacing}

\newcommand{\CNTY}{\texttt{CNTY}}
\newcommand{\CVX}{\texttt{CVX}}
\newcommand{\QCVX}{\texttt{Q-CVX}}
\newcommand{\NI}{\texttt{NI}}
\newcommand{\DPI}{\texttt{DPI}}
\newcommand{\MONO}{\texttt{MONO}}
\newcommand{\AVG}{\texttt{AVG}}
\newcommand{\GAVG}{\texttt{GAVG}}
\newcommand{\MAX}{\texttt{MAX}}
\newcommand{\GMAX}{\texttt{GMAX}}
\newcommand{\MIN}{\texttt{MIN}}

\newcommand{\CCVX}{\texttt{CCVX}}
\newcommand{\CIV}{\texttt{CIV}}

\newcommand{\suport}[1]{\lceil#1\rceil}

\pgfplotsset{compat=newest} 
\pgfplotsset{plot coordinates/math parser=false}
\newlength\figureheight
\newlength\figurewidth
\newcommand{\X}{\mathcal{X}}
\newcommand{\Y}{\mathcal{Y}}

\newcommand{\W}{\mathcal{W}}

\newcommand{\U}{\mathcal{U}}
\newcommand{\Lk}{\mathcal{L}}

\newcommand{\Real}{\mathbb{R}}

\newcommand{\argmin}{\arg\!\min}
\newcommand{\argmax}{\arg\!\max}

\newcommand\Apply {\triangleright}

\newcommand\Dist {\mathbb{D}}

\newcommand\Hyper[2]  {[#1{\Apply}#2]}

\newcommand\Thm[1] {Thm~\ref{#1}}

\newif\ifqif

\qiftrue

\ifqif
        
\else

\fi

\newtheorem{theorem}{Theorem}
\newtheorem{lemma}{Lemma}
\newtheorem{definition}{Definition}

\newtheorem{proposition}{Proposition}

\newtheorem{remark}{Remark}

\begin{document}

\title{An Extension of the Adversarial Threat Model in Quantitative Information Flow}

\author{Mohammad A.~Zarrabian,~\IEEEmembership{Member,~IEEE,} and~Parastoo~Sadeghi,~\IEEEmembership{Senior Member,~IEEE}
\thanks{Mohammad A.~Zarrabian is with the College of Engineering, Computing, and Cybernetics, Australian National University, Canberra, Australia, e-mail: mohammad.zarrabian@anu.edu.au. Parastoo Sadeghi is with the School of Engineering and Technology, the University of New South Wales, Canberra, Australia, e-mail: p.sadeghi@unsw.edu.au. This work was supported by the Australian Research Council Future Fellowship under Grant FT190100429.}}

\maketitle

\begin{abstract}
    In this paper, we propose an extended framework for quantitative information flow (QIF), aligned with the previously proposed core-concave generalization of entropy measures, to include adversaries that use Kolmogorov-Nagumo~\(f\)-mean to infer secrets in a private system. 
    Specifically, in our setting, an adversary uses Kolmogorov-Nagumo~\(f\)-mean to compute its best actions before and after observing the system's randomized outputs. This leads to generalized notions of prior and posterior vulnerability and generalized axiomatic relations that we will derive to elucidate how these \(f\)-mean-based vulnerabilities interact with each other. 
    We demonstrate the usefulness of this framework by showing how some notions of leakage that had been derived outside of the QIF framework and so far seemed incompatible with it are indeed explainable via such an extension of QIF. 
    These leakage measures include \(\alpha\)-leakage, which is the same as Arimoto mutual information of order \(\alpha\), maximal \(\alpha\)-leakage, which is the \(\alpha\)-leakage capacity, and maximal \( (\alpha,\beta) \)-leakage, which is a generalization of the above and captures local differential privacy as a special case. We define the notion of generalized capacity and provide partial results for special classes of functions used in the Kolmogorov-Nagumo mean. We also propose a new pointwise notion of gain function, which we coin pointwise information gain. We show that this pointwise information gain can explain R{\'e}yni divergence and Sibson mutual information of order \(\alpha \in [0,\infty]\) as the Kolmogorov-Nagumo average of the gain with a proper choice of function \(f\).
\end{abstract}


\section{Introduction}
    Information leakage is a main concern in computing and data processing systems. 
    To address this concern, \textit{quantitative information flow} (QIF)~\cite{quantitiveInformationFlow}, starting with the pioneering work
    of Smith~\cite{2009foundationsQIF}, has focused on interpreting privacy and the operational meaning of information leakage in a practical and meaningful way within a framework known as the adversarial threat model~\cite{2012Alvimgain}. 
    Consider a random variable $X,$ representing a secret to be protected from adversarial attacks. 
    The secret can be a database held by a government agency, an individual's unique typing pattern, a password, and so on. 
    To protect $X$, a privacy-enhancing procedure, also known as a mechanism, takes $X$ as the input and produces another random variable,
    denoted by $Y,$ as the system output through a probabilistic
    mapping given by the conditional probability $p_{Y|X}$.
    
    Operationally meaningful measures in QIF include \textit{Bayes vulnerability}~\cite{2009foundationsQIF} (complement of min-entropy) and its generalized version $g$-vulnerability~\cite{2012Alvimgain}, where $g$ is a gain function representing the guessing actions and rewards of the adversary. 
    This approach quantifies the threat as a vulnerability by maximizing the expected gain functions before and after observing randomized data.\footnote{The threat can alternatively be expressed as the minimization of expected loss,  yielding an uncertainty (entropy) measure~\cite{2016alvimaxioms}.} 
    Henceforth, information leakage is defined as the additive or multiplicative difference between the posterior and prior vulnerability~\cite{2014AdditiveMultpliQIF}. 
    A main strength of QIF is that worst-case threats can be quantified by taking the worst case of information leakage over all priors, gain functions, or both. This results in robust notions of leakage \emph{capacity}~\cite{2014AdditiveMultpliQIF}. 
    Two test-of-time awards for~\cite{2009foundationsQIF} and most recently for~\cite{2012Alvimgain} at CSF 2024 are testaments to the success of the QIF framework in the past 15 years. 
    
    In parallel with the QIF, other measures and frameworks to quantify privacy leakage have been developed and studied. Notable examples include differential privacy (DP)~\cite{2006CalibNoiseDFP,2006DFP} and local differential privacy (LDP)~\cite{2011Learnprivately,2013LDPMiniMax}. Connections between DP and LDP notions of privacy and operational quantities from QIF have been established in~\cite{Natash2019comparing,2022EXplainEps,2020MaxL,2023PML}. 
    In particular, it was revealed in~\cite{2022EXplainEps} that LDP is the leakage capacity among all adversaries interested in maximum (worst-case) information leakage across all outputs. 
    
    Information-theoretic privacy measures have also been investigated in the past few years. 
    A notable case is maximal leakage~\cite{2020MaxL}, which has spawned a growing research interest in the development of information-theoretical privacy measures.
    These include pointwise maximal leakage~\cite{2023PML,2023extremalPML,grosse2024quantifying,grosse2024quantifyingjournal}, \(\alpha\)-leakage and its maximal variant~\cite{2019TunMsurInfLeak_PUT,2020alphaproperties,2020OntheRobustness,2023Generalgain, kurri2024maximal}, as well as maximal $(\alpha,\beta)$-leakage~\cite{2022Alphabetleakage,2024GilaniUnifyingMabel}. 
    These measures have been proven useful in applications such as membership privacy~\cite{2021SaraMember,2021FarokhiMember} and machine learning~\cite{2021KurrialphaGan,2022KurriAlphaGan,2022AlphaClass,2024AlphaGAN}.
    
    The authors in~\cite{2019TunMsurInfLeak_PUT} extended the notion of maximal leakage from~\cite{2020MaxL} to \(\alpha\)-leakage and maximal \(\alpha\)-leakage. The \(\alpha\)-leakage measure is reduced to the Arimoto mutual information of order \(\alpha\)~\cite{arimoto1977information}, and maximal \(\alpha\)-leakage is the worst-case \(\alpha\)-leakage over all randomized guesses of the secret and its prior. 
    In particular, for $\alpha = \infty$, maximal \(\alpha\)-leakage becomes maximal leakage.
    Building on this, the authors in~\cite{2024GilaniUnifyingMabel} generalized the concept to maximal $(\alpha,\beta)$-leakage, which encompasses several privacy measures depending on the choice of $(\alpha,\beta)$.
    This includes maximal \(\alpha\)-leakage ($\beta\!=\!1$), maximal leakage ($\beta\!=\!1, \alpha\!=\!\infty$), R\'enyi LDP ($\alpha\!=\!\beta$), and LDP ($\alpha=\beta=\infty$)\footnote{Via vectorization of secrets and conditioning on a vector of secrets as side information, it is also possible to interpret DP as a maximal $(\alpha,\beta)$-leakage.}.
    While the practical applications of $(\alpha,\beta)$-leakage remain largely unexplored, a parallel from maximal $(\alpha,\beta)$-leakage can be drawn to Sharma-Mittal parameterized entropies~\cite{1975SharmaMittal}, which similarly generalize to different entropy measures such as Tsallis entropies~\cite{1988Tsallis}, with broad applications across various fields such as clustering~\cite{koltcov2019estimating} and fuzzy logic~\cite{verma2021sharma}.

   \subsection{Motivation for This Paper}
   
        Most \(\alpha\)-based leakage measures have been introduced following intuitive arguments.
        Despite interesting results and applications, there are still questions about the operational meaning or the adversarial threats such measures represent. We elaborate below.
        
        The first issue concerns what the adversary is guessing in \(\alpha\)-based leakage measures. 
        Instead of using the adversarial gain function model in QIF, the work~\cite{2020MaxL} and follow-up works advocate for a model where a randomized function of secret $X$, denoted by $U$, is guessed by the adversary, subject to the Markov chain $U-X-Y$.
        Despite the seemingly richer adversarial model, maximal leakage becomes identical to the Bayes capacity in QIF. Indeed, both coincide with the Sibson mutual information~\cite{1969SibsonInfRad} of order $\alpha = \infty$. 
        Therefore, a randomized guess model does not encompass any additional risk of information leakage. 
        This has been formally proved in works such as~\cite{2023PML,2022EXplainEps}. 
        See also~\cite{grosse2024quantifyingjournal} for general alphabets and risk-averse adversaries. 

        The second issue is that to our best knowledge, \(\alpha\)-leakage, maximal \(\alpha\)-leakage, and maximal $(\alpha,\beta)$-leakage have not been explained nor interpreted in the QIF framework for \(\alpha\) other than $\alpha=\infty$ or $\alpha=1$ (Shannon mutual information). 
        In other words, except for $\alpha=1, \infty$, no gain (or loss function) within the existing QIF framework is known, which leads to an \(\alpha\)-based leakage measure. 
        It is therefore unclear whether a Markovian randomized guessing model is essential for \(\alpha\)-based leakage measures for finite \(\alpha\).\footnote{We also note technical inconsistencies in some of the definitions and derivations of \(\alpha\)-based leakage measures that are beyond the scope of the Introduction but will be dealt with later in the paper.}

        In this paper, we are concerned with bridging the gap between QIF and leakage measures that have been derived outside of the QIF framework with an aim to extend encompassing features of the QIF framework to include \(\alpha\)-based leakage measures (and more) and to explain them using the consistent and robust language of QIF. 

        The most relevant works to this paper are~\cite{2019GeneEntropySymm,2020CondEntropyAxiom,2021ConvCorConV}, which propose and develop a similar generalized framework for entropy measures known as core-concave entropy.   
        In~\cite{2019GeneEntropySymm}, a generalized entropy was introduced using a generalized mean, which reduces to the Kolmogorov-Nagumo \(f\)-mean in the case of conditional entropy. 
        This formulation unifies various entropy measures, including conditional R\'enyi entropy, Sharma-Mittal entropy, and guessing entropy~\cite{1996IneqGuessing}.  
        The axiomatization of core-concave entropies was further developed in~\cite{2020CondEntropyAxiom,2021ConvCorConV}. Specifically,~\cite{2020CondEntropyAxiom} introduced a set of axioms based on generalized average and minimum conditional entropy, while~\cite{2021ConvCorConV} later resolved the dichotomy between these sets of axioms through a limit construction of generalized average entropies.
            
        The core-concave framework adopts a holistic information-theoretic perspective to establish a consistent generalization. 
        Instead, this paper follows a step-by-step QIF approach, where generalized vulnerabilities are explicitly defined, and for each considered leakage measure, the corresponding gain function and the optimization process are analyzed. 
        These gain functions and optimizations play a crucial role in adversarial modeling and have significant implications for applications such as private machine learning---elements that are bypassed in the core-concave approach. 
        We note that these steps were bypassed in the core-concave approach of~\cite{2019GeneEntropySymm,2020CondEntropyAxiom,2021ConvCorConV}.
        
        Furthermore, in the core-concave framework, the concavity of prior entropy is relaxed to core-concavity, which defines a generalized average of the entropy measure with a strictly increasing function. Since we maintain convexity as a fundamental axiom of prior vulnerability, the implications of our axioms differ from those in the core-concave framework, a distinction that will be clarified in Section~\ref{sec:axioms}.

    \subsection{Contributions and Organization of Results}
    
        Fortunately, it turns out that extended forms of prior and posterior vulnerability using the generalized Kolmogorov-Nagumo \(f\)-mean approach~\cite{2019GeneEntropySymm,2020CondEntropyAxiom,2021ConvCorConV} provide the key to explaining \(\alpha\)-leakage measures, and potentially much more, in the QIF framework. 
        From a high level, the adversary applies a more advanced averaging technique via the Kolmogorov-Nagumo \(f\)-mean (compared to a ``plain vanilla" averaging) to determine its best reward and corresponding action in guessing $X$. This leads to generalized vulnerability and leakage measures. To make these fully compatible with the existing QIF framework, we are also required to prove axiomatic relations for these quantities. 
        Our contributions and organization of the results are listed as follows:
    
        \begin{enumerate}
            \item In Section~\ref{sec:critical}, we will review the original definitions of \(\alpha\)-based leakage measures in previous works and discuss some inconsistencies in their definitions, which have been resolved in this paper.
            
            \item In Section~\ref{Sec:Genralized}, we propose extended forms of prior and posterior vulnerability and leakage using the generalized Kolmogorov-Nagumo \(f\)-mean approach. We also clarify the relationship between these extended forms and the core-concave framework. In the following sections, we use our generalized framework to explain most threat models developed elsewhere.
            
            \item Section~\ref{u:is:redundant} is devoted to showing the applicability of our generalized QIF framework to interpret the operational meaning of \(\alpha\)-based leakage measures, including \(\alpha\)-vulnerability, \(\alpha\)-leakage, maximal \(\alpha\)-leakage, and maximal $(\alpha,\beta)$-leakage. We first prove that guessing a randomized function of $X$, such as $U$ does not add to the adversarial threat model beyond what the generalized capacity in the extended QIF can do. Therefore, the generalized $g$-leakage framework encompasses all such guessing adversaries. This simplifies the operational interpretation of maximal \(\alpha\)-based measures and resolves some inconsistencies discussed in Section~\ref{sec:critical}. Since maximal $(\alpha,\beta)$-leakage encompasses LDP and R\'enyi LDP, our proposed framework also contains these leakage capacities as special cases.  While LDP was already characterized from the lens of QIF in~\cite{2022EXplainEps} as a max-case capacity, this is the first time R\'enyi LDP is explained from the lens of QIF. 
            
            \item We explore the notion of generalized leakage capacity in Section~\ref{sec:capacity}. We present partial results on the generalized leakage capacity for the special class of \(f\)-mean functions with a multiplicative inverse.
            
            \item In Section~\ref{Sec:Sibson}, we complete our interpretation of \(\alpha\)-based measures by proposing a new information gain function, which we coin \textit{pointwise information gain}. This pointwise information gain is inspired by the  R\'enyi's information gain~\cite{1961measuresreny}. Using this gain function, we interpret the operational meaning of R{\'e}nyi divergence as the pointwise \(\alpha\)-leakage and Sibson mutual information\cite{1969SibsonInfRad} for all $\alpha \in [0,\infty]$ for the first time, to the best of our knowledge. 
            
            \item Section~\ref{sec:axioms} reviews the existing axioms for prior and posterior vulnerabilities in QIF~\cite{2016alvimaxioms}, as well as the generalized axioms proposed in~\cite{2020CondEntropyAxiom,2021ConvCorConV}. We demonstrate that the conventional QIF axioms are satisfied by our proposed generalized vulnerabilities, following the approach of~\cite{2016alvimaxioms}, and clarify the specific relaxations considered in comparison to~\cite{2020CondEntropyAxiom,2021ConvCorConV}. The conclusions are presented in Section~\ref{sec:conclusion}.
        \end{enumerate}

\section{Background}\label{sec:background}

    This section provides a concise overview of key concepts and terminology of QIF as outlined in~\cite{2016alvimaxioms}, and {\(\alpha\)-information measures}~\cite{1961measuresreny, 1969SibsonInfRad, arimoto1977information, 1995GeneralizedCuttoff} which will be used throughout this paper. For more in-depth explanations, readers are referred to~\cite{2016alvimaxioms,quantitiveInformationFlow,2021ErrorExponentAlphaInfr} and the references therein.
    
    The \textit{secret} \( X \) represents the information that must be protected from adversaries who know \( X \) only via a prior probability distribution 
    \(\pi\) defined over the alphabet \(\mathcal{X}\). A system is characterized by the triple \( (\mathcal{X}, \mathcal{Y}, C) \), where \( \mathcal{X} \) and \( \mathcal{Y} \) are finite sets of input and output symbols, respectively, and \( C = P_{Y|X} \) is a channel matrix of size \( |\mathcal{X}| \times |\mathcal{Y}| \). The elements of matrix \( C \) represent the conditional probability \( C_{x,y} = \Pr[Y=y|X=x] = p(y|x) \), denoting the likelihood of observing output \( y \) given the input \( x \). Each row of \( C \) is a probability distribution over \( \mathcal{Y} \), with all elements being non-negative and summing to 1. 
    
    It is typically assumed that in addition to the prior distribution $\pi$, the adversary is also aware of the channel. Consequently, the adversary can compute the joint distribution \(p(x,y)\!\!=\!\pi_x C_{x,y} \), marginals \(\pi_{x}\!\!=\!\sum_{y\in\Y} p(x,y)\) and \(p(y)\!=\!\!\sum_{x\in\X}p(x,y) \), as well as posteriors \(\delta_{x}^{y}\!\!=\!\frac{p(x,y)}{p(y)}\), when $p(y)\!\neq\! 0$. The channel's function, therefore, is to update the adversary’s knowledge about $X$ from the prior \(\pi\) to a set of posterior distributions \(\delta^{y}\!\!=\!\!p(X|y)\), each occurring with probability $p(y)$.
    
    Let \( \mathbb{D}\mathcal{X} \) represents the set of distributions over \( \mathcal{X} \), and \( \suport{\pi} \) denote the support of $\pi$.
    The pair \( [\pi, C] \) yields the posterior \( \delta^{y} \) and corresponding \( p(y) \) for each \( y \in \mathcal{Y} \), which are referred to as the inner and outer distributions, respectively. Instead of treating \( p_{Y} \) as a distribution on \( \mathcal{Y} \), it can be viewed as a distribution over the posteriors \( \delta^{y}\). This creates a distribution over distributions, denoted by \( \mathbb{D}(\mathbb{D}\mathcal{X}) = \mathbb{D}^2\mathcal{X} \), also known as a hyper distribution. Let \( \Delta \) represent a general hyper-distribution, and \( [\pi, C] \) the hyper-distribution resulting from the channel \( C \) acting on the prior \( \pi \). The support of the hyper-distribution is denoted by \( \suport{\Delta} \), and \( [\pi] \) indicates a point hyper that assigns probability 1 to \( \pi \).

    \subsection{Definitions and Results from QIF}
    
        \begin{definition}
            For a set of (possibly infinite) guesses \( \mathcal{W} \)  that an adversary might make about \( X \), the gain function $g:\!\mathcal{W}\!\times\!\mathcal{X}\!\rightarrow\!\mathbb{R}$ measures the adversary's expected gain for a guess \( w \) when the actual secret value is \( x \). 
            The \( g \)-vulnerability function assesses the threat by calculating the adversary's expected gain for an optimal choice of \( w \). 
            Accordingly, the threat by prior distribution is given by the \textbf{prior $g$-vulnerability} as:
            \begin{align*}
                V_g(\pi) =\sup_{w \in \mathcal{W}} \sum_{x \in \mathcal{X}} \pi_x g(w,x).
            \end{align*}
            Moreover, the threat after observation of $Y$ is given by the posterior $g$-vulnerability. The first class of posterior $g$-vulnerability is the average (\AVG), which is defined as
            \begin{align*}
                \widehat{V}_g[\pi,C] &= \sum_{y\in\Y} \sup_{w\in\W} \sum_{x\in\X} \pi_x C_{x,y} g(w,x) \nonumber\\
                &= \sum_{y\in\Y} p(y) \sup_{w\in\W} \sum_{x\in\X} \delta_{x}^{y} g(w,x) 
                =\sum_{y\in\Y}p(y)V_{g}(\delta^{y}).
            \end{align*}
            The second class is the maximum (\MAX), which is given by 
            \begin{align*}
                \widehat{V}^{\max}_g[\pi,C] &= \max_{y\in \Y} V_{g}(\delta^{y}),
            \end{align*}
            and measures the worst-case posterior threat.
        \end{definition}
        In reference to~\cite{quantitiveInformationFlow}, we allow for negative values in the function $g(w,x)$ to indicate ``losses'' for guessing $w$ when the secret is $x$. However, it is necessary for the expected gain $V_g$ always to be non-negative so that a value of zero indicates no vulnerability. Therefore, at least one positive value should be in the co-domain of $g(w,x)$.
        An alternative way to assess a threat is through \textit{uncertainty}, which is defined based on a loss function that the adversary aims to minimize. Accordingly, the prior uncertainty is defined as \[U_{l}(\pi)=\inf_{w\in\W}\sum_{x\in\X}\pi_{x}l(w,x),\] where $l(w,x)$ is the loss function. 
    
        \begin{definition} 
            Leakage measures can be either additive or multiplicative, which are defined as follows:
            \begin{equation}
                \begin{array}{rlr}
                &\text{Additive:~~~~~~~} 
                \mathcal{L}_{g}^{+}(\pi,C)=\widehat{V}_{g}[\pi, C]-V_{g}(\pi),\nonumber \\
                &\text {Multiplicative:~} 
                \mathcal{L}_{g}^{\times}(\pi,C) =\log \left(\widehat{V}_{g}[\pi, C] / V_{g}(\pi)\right).
                \end{array}
            \end{equation}
            One can replace $\widehat{V}_{g}[\pi, C]$ with $\widehat{V}^{\max}_g[\pi,C]$, resulting in \textit{max-case $g$-leakage}~\cite{2022EXplainEps}. In this paper, we consider only multiplicative max-case $g$-leakage denoted by $\mathcal{L}_{g}^{\max}(\pi,C)$.
        \end{definition}

        Another notable measure is leakage capacity, which serves as a measure of the robustness of leakage by maximizing it over the prior $\pi$, the gain function $g$, or both. These maximizations account for our potential uncertainty regarding the prior knowledge or adversary's chosen gain function. In essence, capacities represent universal quantities that capture worst-case scenarios, defining the boundaries of maximum threat. There are six interpretations of capacity in total (three for each definition); however, since we focus exclusively on multiplicative leakage, we review results related to $\mathcal{L}_{g}^{\times}(\pi,C)$ only. 
        The three definitions of multiplicative capacity are:
        \begin{itemize}
            \item Supremum over prior $\pi:\mathbb{D}\mathcal{X}$ with fixed $g$: $$\mathcal{L}_{g}^{\times}(\forall,C)=\sup_{\pi}\mathcal{L}_{g}^{\times}(\pi,C);$$
            \item Supremum over $g$ with fixed $\pi$: $$\mathcal{L}_{\forall}^{\times}(\pi,C)=\sup_{g}\mathcal{L}_{g}^{\times}(\pi,C);$$
            \item Supremum over both $\pi$ and $g$: $$\mathcal{L}_{\forall}^{\times}(\forall,C)=\sup_{\pi}\mathcal{L}_{{\forall}}^{\times}(\pi,C)=\sup_{g}\mathcal{L}_{g}^{\times}(\forall,C).$$
        \end{itemize}
        
        For the class of non-negative gain functions $g$, the following results characterizes $\mathcal{L}_{\forall}^{\times}(\forall,C)$.
        \begin{theorem}
            $\mathcal{L}_{\forall}^{\times}(\forall,C)$ is given by the Bayes capacity $\mathcal{ML}(C)$
            \begin{align*}
                \mathcal{L}_{\forall}^{\times}(\forall,C)=\mathcal{ML}(C) = \log \sum_{y\in\Y}\max_{x\in\X}C_{x,y},
            \end{align*}
            where the supremum is achieved by uniform prior and identity gain function~\cite{2009Quantitaiveleakage,2012Alvimgain}. 
        \end{theorem}
        
        Bayes capacity is known as maximal leakage in information theory~\cite{2020MaxL}, and both are equal to Sibson mutual information of order $\alpha = \infty$. Another important measure is the LDP leakage, the max-case capacity among all adversaries interested in maximum information leakage across all outputs~\cite{2022EXplainEps}.
        \begin{theorem} 
            $\mathcal{L}_{\forall}^{\max}(\forall,C)$ is given by LDP leakage $\Lk^{\text{LDP}}(C)$.
            \begin{align*}
                \mathcal{L}_{\forall}^{\max}(\forall,C)&=\sup_{\pi,g}\mathcal{L}_{g}^{\max}(\pi,C)=\sup_{\pi,g}\log\frac{V_{g}^{\max}[\pi,C]}{V_{g}(\pi)} \nonumber\\
                &=\log\max_{y\in\Y}\frac{\max_{x\in\X}C_{x,y}}{\min_{x\in\X}C_{x,y}}=\Lk^{\text{LDP}}(C).
            \end{align*}      
        \end{theorem}
        
    \subsection{Definitions of \texorpdfstring{$\alpha$}{alpha}-based Information-theoretic Measures} 
        
        \begin{definition}
            For a given distribution $\pi:\mathbb{D}\X$, the R{\'e}nyi entropy of order $\alpha \in [0,\infty]$ is defined as:
            \begin{align*}
                H_{\alpha}(\pi)=\frac{1}{1-\alpha}\log\sum_{x \in \X}\pi_{x}^{\alpha}=\frac{\alpha}{1-\alpha}\log\|\pi\|_{\alpha}.
            \end{align*}
            
            Let $\mu:\mathbb{D}\X$ be another distribution over $\X$. R{\'e}nyi divergence of order $\alpha \in [0,\infty]$ between $\mu$ and $\pi$ is defined as:
            \begin{align*}
                D_{\alpha}(\mu\|\pi)=\frac{1}{\alpha-1}\log \left(\sum_{x\in\X} \frac{(\mu_{x})^{\alpha}}{(\pi_{x})^{\alpha-1}}\right).
            \end{align*}
        \end{definition}
        We have used the range $[0,\infty]$ for \(\alpha\) since both quantities above are defined by their continuous extensions for $\alpha=0,1,\infty$.
        R{\'e}nyi entropy and divergence of order $\alpha=1$ are Shannon entropy and Kullback-Leibler divergence, respectively. Moreover, for $\alpha=\infty$, R{\'e}nyi entropy is the min-entropy $H_{\infty}(\pi)=-\log V_{b}(\pi)$, where $V_{\text{b}}(\pi)=\max_{x\in \X}\pi_{x}$ is known as Bayes vulnerability~\cite{2009foundationsQIF}.

        \begin{definition}
            For a given hyper $\Delta=[\pi,C]$ with inner $\delta^{y}$ and outer $p(y)$ for each $y\in \Y$, Arimoto mutual information of order $\alpha \in [0,\infty]$ is defined as follows:
            \begin{align*}
                I_{\alpha}^{A}(X;Y) 
                &=H_{\alpha}(\pi) - H_{\alpha}(X|Y)\\
                &=\frac{\alpha}{\alpha-1}\log \frac{\sum_{y\in\Y}p(y)\Big(\sum_{x\in\X}(\delta_{x}^{y})^{\alpha}\Big)^{\frac{1}{\alpha}}}{\Big(\sum_{x\in\X}(\pi_{x})^{\alpha}\Big)^{\frac{1}{\alpha}}},
            \end{align*} 
            where $H_{\alpha}(X|Y)$ is Arimoto conditional entropy of $X$ given $Y$ and is defined as:
            \begin{align*}
                H_{\alpha}(X|Y)=\frac{\alpha}{1-\alpha}\log \sum_{y\in\Y}p(y) \Big( \sum_{x\in\X}(\delta_{x}^{y})^{\alpha}\Big)^{\frac{1}{\alpha}}.
            \end{align*}
            
            Another \(\alpha\)-based leakage measure is Sibson mutual information of order $\alpha \in [0,\infty]$ that  is given by:
            \begin{align*}
                I_{\alpha}^{S} (X;Y)=\frac{\alpha}{\alpha-1}\log\sum_{y\in\Y}\Big(\sum_{x\in\X}\pi_{x}\left(C_{x,y}\right)^{\alpha} \Big)^{\frac{1}{\alpha}}.
            \end{align*}
        \end{definition} 
        
        Both Sibson and Arimoto mutual information give Shannon mutual information for $\alpha=1$. 
        Another notable order is $\alpha=\infty$, where Sibson mutual information reduces to maximal leakage or Bayes capacity:
        \begin{align*}
            I_{\infty}^{S}(X;Y)=\mathcal{ML}(C)=\log\sum_{y\in\Y}\max_{x\in\X}C_{x,y};
        \end{align*}
        and Arimoto mutual information will become 
        \begin{align*}
            I_{\infty}^{A}(X;Y)=\log\frac{\sum_{y\in\Y}p(y)\max_{x\in\X}\delta_{x}^y}{\max_{x\in\X}\pi_{x}}.
        \end{align*} 
        
        It has been proved that Arimoto and Sibson mutual information has the same supremum over $\pi$~\cite[Thm. 5]{2015alphaMI}, 
        \begin{equation*}
            \sup_{\pi}I_{\alpha}^{A}(X;Y)=\sup_{\pi}I_{\alpha}^{S}(X;Y).  
        \end{equation*}
        Despite recovering Shannon mutual information when $\alpha = 1$, these measures are not symmetric in general: $I_{\alpha}^{S}(X;Y)\neq I_{\alpha}^{S}(Y;X)$ and $I_{\alpha}^{A}(X;Y)\neq I_{\alpha}^{A}(Y;X)$.

\section{A Critical Review of Previous Derivations of \texorpdfstring{$\alpha$}{alpha}-Based Leakage Measures}\label{sec:critical}

    In this section, we review the original definitions of maximal leakage and \(\alpha\)-based measures and discuss some inconsistencies in those definitions. 
    Subsequent to $g$-leakage in QIF, maximal leakage~\cite{2020MaxL} was proposed as an alternative to characterizing adversarial threats. 
    In this scenario, a randomized function of secret $X$, denoted by $U$ is to be guessed over the alphabet $\U$.  The leakage is defined as the supremum over all $U$ and $\hat{U}$ functions subject to the Markov chain $U-X-Y-\hat{U}$, where $\hat{U}$ is the outcome of the guess over the same alphabet $\U$. The maximal leakage is given by:
        \begin{align}
            &\mathcal{L}_{\max}(X\!\rightarrow\!Y)=\sup_{U-X-Y-\hat{U}}\log \frac{\Pr(U=\hat{U})}{\max_{u\in \U}p(u)} \nonumber \\
            &=\sup_{U-X-Y}\log \frac{\sum_{y\in\Y}\max_{u}p(u,y)}{\max_{u\in \U}p(u)} 
            =\mathcal{ML}(C). \label{eq:maxLeak}
        \end{align}
    
    The \(\alpha\)-leakage and maximal \(\alpha\)-leakage have been defined in~\cite{2019TunMsurInfLeak_PUT} in connection with  \(\alpha\)-loss for $\alpha \in [1,\infty]$. 
    Subsequently, its definition was modified in~\cite{2020alphaproperties} to extend the range of \(\alpha\) to $(0,1)\cup [1,\infty]$, which is given as follows:

    \begin{definition}[\(\alpha\)-loss~\cite{2019TunMsurInfLeak_PUT,2020alphaproperties}]
        For a probabilistic estimator $\hat{\pi}:\mathbb{D}\X$ and a parameter $\alpha>0$, the \(\alpha\)-loss is given by:
        \begin{align}\label{eq:alphaloss}
            \hspace{-7pt}\ell_{\alpha}(\hat{\pi}_{x})\triangleq
            \begin{cases}
                \frac{\alpha}{\alpha-1}\left(1-(\hat{\pi}_{x})^{\frac{\alpha-1}{\alpha}}\right), &\alpha\in(0,1)\cup(1,\infty), \\
                \log \frac{1}{\hat{\pi}_{x}}, & \alpha=1, \\
                1-\hat{\pi}_{x}, &\alpha=\infty.
            \end{cases}
        \end{align}
    \end{definition}
    \begin{remark}\label{remark:alphalossex}
        For a given distribution \(\pi:\mathbb{D}\X\), the minimum expected \(\alpha\)-loss is given by:
        \begin{equation*}
             \min_{\hat{\pi}:\mathbb{D}\X} \sum_{x\in\X}{\pi}_{x}\ell_{\alpha}(\hat{\pi}_{x})
            =\frac{\alpha}{\alpha-1} \bigg( 1-\exp{\Big(\frac{1-\alpha}{\alpha}H_{\alpha}(\pi)\Big)} \bigg), 
        \end{equation*}
        with the optimal answer $\hat{\pi}_{x}^{*}=\frac{(\pi_{x})^{\alpha}}{\sum_{x\in\X}(\pi_{x})^{\alpha}}$.
    \end{remark}   
        
    \begin{definition}[\(\alpha\)-leakage and maximal \(\alpha\)-leakage~\cite{2019TunMsurInfLeak_PUT}] \label{def:alphaleak} 
        For a given joint distribution $p(X,Y)$, given \(\alpha\)-loss in~\eqref{eq:alphaloss}, and $\alpha \in [0,\infty]$ the \(\alpha\)-leakage from $X$ to $Y$ is defined as:
        \begin{align}\label{eq:alphaleak}
            \mathcal{L}_{\alpha}(X\!\!\rightarrow\!\!Y) 
            &\triangleq \frac{\alpha}{\alpha-1}\log\frac{\displaystyle \frac{\alpha}{\alpha-1}\!-\!\min_{\hat{\delta^{y}:\mathbb{D}\X}}\sum_{y\in\Y}\!p(y)\!\sum_{x\in\X}\delta_{x}^{y}\ell_{\alpha}(\hat{\delta}_{x}^{y})}{\displaystyle \frac{\alpha}{\alpha-1}-\min_{\hat{\pi}:\mathbb{D}\X}\sum_{x\in\X}\pi_{x}\ell_{\alpha}(\hat{\pi}_{x})}\nonumber\\
            &=I_{\alpha}^{A}(X;Y).
        \end{align}
        In~\eqref{eq:alphaleak}, $\hat{\pi}$ and  $\hat{\delta}^{y}$ are the prior and posterior estimators.
    \end{definition}
    \begin{remark}[{Maximal \(\alpha\)-leakage}]
        For a randomized function $U$ of secret $X$, the maximal \(\alpha\)-leakage is defined as:
        \begin{align}
            &\mathcal{L}^{max}_{\alpha}(X\!\!\rightarrow\!\!Y)\triangleq \sup_{U-X-Y} \mathcal{L}_{\alpha}(U\!\!\rightarrow\!\!Y)  \nonumber\\
            &=\begin{cases}\label{eq:Xmaximalleak}
                \sup_{{\pi}}I_{\alpha}^{A}(X;Y)=\sup_{{\pi}}I_{\alpha}^{S}(X;Y), & \alpha\neq 1 \\
                I(X;Y), &\alpha=1 ,
            \end{cases}
        \end{align} 
        where $ \mathcal{L}_{\alpha}(U\!\!\rightarrow\!\!Y)$ is given by~\eqref{eq:alphaleak} where all distributions over $\X$ are replaced by the ones over $\U$. 
    \end{remark}
       
    \begin{definition}[Maximal ($\alpha,\beta$)-leakage~\cite{2024GilaniUnifyingMabel}] 
        Given a hyper distribution $\Delta=[\pi,C]$ and a randomized function $U$ with probabilistic prior and posterior estimators $\hat{p}_{U}, \hat{\delta}^{y}:\mathbb{D}\U$,  respectively, the maximal $(\alpha,\beta)$-leakage from $X$ to $Y$ is defined as:
        \begin{align}
            &\mathcal{L}_{\alpha,\beta}(X\!\!\rightarrow\!\!Y)\triangleq 
            \sup_{\pi}\sup_{U-X-Y}\frac{\alpha}{\alpha-1} \nonumber \\
            &\log \frac{\displaystyle \max_{\hat{\delta}^{y}:\mathbb{D}\U}\left[\sum_{y\in\Y}p(y)\left(\sum_{u\in\U}p(u|y)(\hat{\delta}^{y}_{u})^{\frac{\alpha-1}{\alpha}} \right)^{\beta}\right]^{\frac{1}{\beta}}}{\displaystyle \max_{\hat{p}_{U}:\mathbb{D}\U}\sum_{u\in\U}p(u)\hat{p}(u)^{\frac{\alpha-1}{\alpha}}} \label{eq:alphabetadef}\\
            &=\max_{x'}\sup_{\Tilde{\pi}}\frac{\alpha}{(\alpha-1)\beta}\log\sum_{y\in\Y}C_{x',y}^{1-\beta}\left( \sum_{x\in\X}\Tilde{\pi}_{x}C_{x,y}^{\alpha} \right)^{\frac{\beta}{\alpha}},\label{eq:alphabetaorigin}
        \end{align} 
        where $\Tilde{\pi}$ is a probability distribution over $\X$ given by $$\Tilde{\pi}_{x}=\frac{\sum_{u\in\U}p(u)^{\alpha}p(x|u)}{\sum_{u\in\U}p(u)^{\alpha}}.$$
        In special cases, this measure represents maximal \(\alpha\)-leakage, maximal leakage,  R{\'e}yni LDP, and LDP (see also Remark~\ref{remark:specialcasealphabeta}). 
    \end{definition}
    These definitions have been valuable in advancing our knowledge of privacy leakage measures. 
    However, they suffer from some inconsistencies that have been resolved in this paper. The main issues are outlined below.
    \begin{itemize}
        \item \textbf{Inconsistency with QIF framework:}
            The primary issue with the above definitions is that they are inconsistent with the QIF framework, even though they have been modeled somewhat similar to the QIF framework; using the ratio of the maximum posterior guessing gain to the prior guessing gain. Specifically, no gain or vulnerability function properly interprets these leakages within the $g$-leakage framework.
            This inconsistency leads to other issues. For example, the coefficient $\frac{\alpha}{\alpha-1}$ leading the logarithm in~\eqref{eq:alphaleak} and~\eqref{eq:alphabetadef},  lack a rigorous justification, despite being intuitively correct.

        \item \textbf{Issues with $\sup_{U-X-Y}$ in maximal measures:}
            The expression $\sup_{U-X-Y}$ in maximal measures is problematic. First, it suggests that guessing a randomized function of $X$ (such as $U$) is necessary to address worst-case adversaries. However, the resulting capacity measure does not introduce any additional risk of information leakage. Second, the model complexity has led to a different interpretations and derivations in the literature. 
            In~\cite{2020MaxL}, it is replaced by $\sup_{\pi}\sup_{p_{U|X}}$, while in~\eqref{eq:alphabetadef} there is an extra $\sup_{\pi}$ that results in $\max_{x'}$, as well as  $\sup_{\Tilde{\pi}}$ in~\eqref{eq:alphabetaorigin}. 
            This may also account for the inconsistency in~\eqref{eq:Xmaximalleak}, where for $\alpha\neq 1$, we have capacities, but for $\alpha=1$, it is only a leakage measure. Note that for $\alpha=1$ in~\eqref{eq:Xmaximalleak}, there is no supremum over $\pi$.

            \item \textbf{Exclusion of the $\alpha = 0$ case:}
             While all \(\alpha\)-measures in information theory are consistently defined over the entire range of $\alpha \in [0, \infty]$, the \(\alpha\)-loss and $\mathcal{L}_{\alpha}(X\!\!\rightarrow\!\!Y)$ were initially introduced for $\alpha \in [1,\infty]$ and later extended to $\alpha \in (0,\infty)$ in~\cite{2020alphaproperties} through an intuitive approach. However, the $\alpha = 0$ case remained excluded, as the definitions cannot be continuously extended to this value.

        \item\textbf{Unclear relationship between \(\alpha\)-loss and \(\alpha\)-leakage:}
        The relationship between \(\alpha\)-loss and \(\alpha\)-leakage remains unclear. Although an uncertainty measure can be defined based on expected \(\alpha\)-loss in Remark~\eqref{remark:alphalossex} with the prior uncertainty measure given by $U_{\ell_{\alpha}}(\pi)=\frac{\alpha}{\alpha-1}(1-\exp(\frac{1-\alpha}{\alpha}H_{\alpha}(\pi)))$, this measure cannot be employed to define a leakage measure.  Furthermore, for $\alpha\!<\!1$, it lacks concavity, violating the axioms of uncertainty measures (See Section~\ref{sec:axioms}).
        In~\cite{ding2024cross}, an \(f\)-mean approach was proposed to bridge these gaps by defining \(\alpha\)-loss and \(\alpha\)-leakage through a new formulation of cross-entropy. However, as highlighted in Remark~\ref{remark:Niloss}, this approach is not consistent with the original QIF framework.      
    \end{itemize}

\section{Generalized Vulnerability and Leakage}\label{Sec:Genralized}

    This section proposes generalized vulnerability, leakage, and capacity measures using the Kolmogorov–Nagumo mean~\cite{1952inequalities} (quasi-arithmetic mean or generalized \(f\)-mean). In the next three sections, we will demonstrate the usefulness of the proposed generalized measures in explaining seemingly incompatible \(\alpha\)-based information leakage measures developed outside of QIF in the new generalized framework. In Section~\ref{sec:axioms}, we will prove the axioms of vulnerability are satisfied for the generalized versions. 
    
    \begin{definition}[\textbf{Kolmogorov–Nagumo mean}]\label{def:Kolmogorov–Nagumo average}
        \emph{Given a set of real numbers $t=\{t_1, t_2, \ldots, t_n\}$ with corresponding weights $\omega_1, \omega_2, \ldots, \omega_n$, where $\omega_k > 0$ and $\sum_{k=1}^{n}\omega_k=1$, the general form of a mean value is expressed as:
        \begin{equation*}\label{eq:Kol-Nag-mean}
            \Bar{t}=f^{-1}\Big(\sum_{k=1}^{n}\omega_kf(t_k)\Big),
        \end{equation*}
        where \(f\) is a strictly monotonic and continuous function with the inverse function of $f^{-1}$.}
    \end{definition}

    \begin{definition}[\textbf{Generalized Prior Vulnerability}]\label{def:generalprior}
        For a given prior $\pi\!:\!\Dist\X$, a gain function $g: \W \times \X \rightarrow \Real$, and any strictly monotonic and continuous function \(f\) with a \textbf{convex} inverse $f^{-1}$, the generalized prior vulnerability is defined as:
        \begin{align}\label{eq:G(PX)}
            V_{f,g}(\pi):=\sup_{w\in\W}f^{-1}\Big(\sum_{x\in\X}\pi_{x} f\big(g(w,x)\big)\Big).
        \end{align} 
         If \(f\) is affine, i.e., $f(t)=at+b$, then, $ V_{f,g}(\pi)=V_{g}(\pi).$ Note that the convexity of $f^{-1}$ implies that \(f\) is either \textbf{convex} and \textbf{decreasing} or \textbf{concave} and \textbf{increasing}. This assumption is sufficient for the axioms of prior vulnerability for  $V_{f,g}(\pi)$ (See Sec.~\ref{sec:axioms}).
    \end{definition}
   
    \begin{definition}[\textbf{Generalized Average Posterior Vulnerability}]\label{def:poset-vulnerability}
        For a hyper $\Delta=[\pi,C]$ and each $y \in \Y$, the generalized vulnerability of each inner $\delta^{y}$ is given by:
        \begin{align}
            V_{f,g}(\delta^{y}) &= \sup_{w\in\W}f^{-1}\bigg(\sum_{x\in\X}\delta_{x}^{y} f\big(g(w,x)\big)\bigg). \label{eq:V(deltay)}    
        \end{align}
        According to~\eqref{eq:V(deltay)}, the generalized average posterior vulnerability is defined as:
        \begin{align}\label{eq:h(t)posterior}
            \widehat{V}_{h,f,g}[\pi, C]
            &:=h^{-1}\bigg(\sum_{y\in\Y}p(y)h\big(V_{f,g}(\delta^{y})\big) \bigg),
        \end{align} 
        where $h$ is a strictly monotonic and continuous function that could potentially be different from \(f\).  
        If $h\neq f$, then we assume it is \textbf{convex} and \textbf{increasing} or \textbf{concave} and \textbf{decreasing}. 
        These assumptions are sufficient for the \DPI~axiom for $\widehat{V}_{h,f,g}[\pi,C]$ (See Sec.~\ref{sec:axioms}). 
         If $h$ is affine, i.e., $h(t)=at+b$, then we have:
        \begin{align}
            \widehat{V}_{h,f,g}[\pi, C]&=\sum_{y\in\Y}p(y)V_{f,g}(\delta^{y})\label{eq:p(y)posetrior}.
        \end{align}
    \end{definition}
    
    \begin{remark}
        For the special case of $h=f$, we have:
        \begin{align}
          &\widehat{V}_{f,f,g}[\pi,C] \nonumber \\
          &=f^{-1}\left(\sum_{y\in\Y}p(y) f\bigg(\sup_{w\in\W} f^{-1} \Big(\sum_{x\in\X}\delta_{x}^{y}f\big(g(w,x)\big)\Big)\bigg)\right)\nonumber\\
            &=\begin{cases}\label{eq:h=f}
                \displaystyle f^{-1}\Big(\sum_{y\in\Y} \sup_{w\in\W} \sum_{x\in\X}C_{x,y}\pi_{x} f\big( g(w, x)\big)\!\Big), &  \!\!f^{-1}\text{increasing},\\
                \displaystyle f^{-1}\Big( \sum_{y\in\Y}  \inf_{w\in\W} \sum_{x\in\X} C_{x,y}\pi_{x} f\big( g(w, x)\big)\!\Big), & \!\!f^{-1}\text{decreasing}.
            \end{cases}
        \end{align}
        When $f^{-1}$ is increasing, we can move the $\sup_{w\in\W}$ inside the function. Then \(f\) and $f^{-1}$ cancel each other. 
        For decreasing $f^{-1}$, $\sup_{w\in\W}$ becomes $\inf_{w\in\W}$ when it is moved inside $f^{-1}$.
    \end{remark}
      
    \begin{definition}[\textbf{Generalized Max Posterior Vulnerability}]
        Generalized maximum posterior vulnerability is defined as:
        \begin{align}
            \widehat{V}^{\max}_{f,g}[\pi,C]=\max_{y\in\Y}V_{f,g}(\delta^{y}).
        \end{align}
    \end{definition}

    \begin{definition} 
        Similar to~\cite{2016alvimaxioms}, the generalized forms of leakage measures are defined as:
        \begin{equation}
            \begin{array}{rlr}
                &\text{Additive:~~\quad\quad} 
                \mathcal{L}_{h,f,g}^{+}(\pi,C)=\widehat{V}_{h,f,g}[\pi, C]-V_{f,g}(\pi),\nonumber \\
                &\text {Multiplicative:~} 
                \mathcal{L}_{h,f,g}^{\times}(\pi,C) =\log\left({\widehat{V}_{h,f,g}[\pi, C]}/{V_{f,g}(\pi)}\!\right).
            \end{array}
        \end{equation}
    \end{definition}
        
    If we replace $\widehat{V}_{h,f,g}[\pi, C]$ with $\widehat{V}^{\max}_{f,g}[\pi,C]$, we obtain the generalized max-case leakage.
    Accordingly, the generalized capacities are the supremum of generalized leakage over $\pi$, $g$, or both for fixed \(f\) and $h$. The generalized multiplicative capacities are denoted by
    \begin{align}\label{eq:GenCap}
         \mathcal{L}_{h,f,g}^{\times}(\forall,C),~~ \mathcal{L}_{h,f,\forall}^{\times}(\pi,C),~~  \mathcal{L}_{h,f,\forall}^{\times}(\forall,C).
    \end{align}

    We will study some of these quantities and their applications in  Sections~\ref{u:is:redundant} to~\ref{sec:capacity}.

    \subsection{Relation to Core-Concave Generalized Framework}

        Here, we discuss the relationship between our generalized definitions, and the framework proposed in~\cite{2019GeneEntropySymm,2020CondEntropyAxiom,2021ConvCorConV}. 

        \begin{definition}[{Def.1~\cite{2020CondEntropyAxiom}}]
            A \textit{core-concave entropy} $H\!=\!(\eta,F)$ is a pair such that:
            \begin{enumerate}
                \item \(f\) is a real-valued function over an $n$-dimensional simplex $\Lambda_{n}$ that is continuous and concave;
                \item $\eta$ is a continuous and strictly increasing real-valued function defined over the image of \(f\).
            \end{enumerate}
        \end{definition}
        According to this definition, a general form of entropy is given by $H(\pi)=\eta(F(\pi))$. 
        This definition captures most of the entropy measures in the literature. 
        While different choices of $(\eta, F)$ can result in the same entropy measure, the conditional form of entropy determines $(\eta,F)$ uniquely up to a linear transformation~\cite[Thm. 1]{2020CondEntropyAxiom}. 
            
        \begin{definition}[{Def. 2~\cite{2020CondEntropyAxiom}}]
            Given a core-concave entropy $H =(\eta,F)$, its “conditional” form is defined as:
            \begin{align}\label{eq:condEtaF}
                H(X|Y)=\eta\bigg( \sum_{y\in \Y^{+}}p(y)F(X|y)\bigg),
            \end{align}
            where $\Y^{+}$ is the support of $Y$ and $F(X|y)$ is shorthand for  $F(p_{X|y})$.
            In terms of the (unconditional) entropy,~\eqref{eq:condEtaF} is equivalent to:
            \begin{align}
                 H(X|Y)=\eta\bigg(\sum_{y\in \Y^{+}}p(y)\eta^{-1}\Big(H(X|y)\Big)\bigg),
            \end{align}
            which has a Kolmogorov–Nagumo form.
        \end{definition}   
        A generalized core-convex vulnerability can similarly be defined by a pair $(\eta, F)$, where \(f\) is convex. Consequently, we can map our generalized vulnerability definitions to the core-convex framework as follows:
        \begin{align*}
            &\eta=h^{-1}, &&\eta^{-1}=h,\\
            &F(X|y)=h\big(V_{f,g}(\delta^{y})\big),  &&h^{-1}\big(F(X|y)\big)=V_{f,g}(\delta^{y}).
        \end{align*}
            
        While our framework employs a generalized mean to define $V_{f,g}(\pi)$, the core-convex framework bypasses these intermediate steps. 
        For instance, for R\'enyi entropy, $F(\pi) = -||\pi||_{\alpha}$ represents the final result derived 
        in our Theorem~\ref{theo:alpha-vunle}, which incorporates a gain function, an \(f\)-mean function, and an optimization process.
        Moreover, while the core-convex framework relaxes the convexity axiom to core-convexity, requiring only that $h$ be increasing, 
        we adhere to convexity and impose additional conditions on $h$. This will be explained further in Section~\ref{sec:axioms}.

\section{Bringing \texorpdfstring{$\alpha$}{alpha}-based Leakage Measures into QIF}\label{u:is:redundant}

    In this section, we interpret \(\alpha\)-vulnerability (this resolves the issues with the R{\'e}nyi entropy), \(\alpha\)-leakage, maximal \(\alpha\)-leakage, and maximal $(\alpha,\beta)$-leakage within the generalized framework. 
    The two latter have been proposed in the maximal leakage framework.  
    In Section~\ref{u:is:redundantA}, we propose generalized maximal leakage and show that it is equivalent to the generalized capacity $\mathcal{L}_{h,f,g}^{\times}(\forall,C)$ given in~\eqref{eq:GenCap}. 
    Then, in Section~\ref{sec:interpret}, we interpret the above mentioned \(\alpha\)-measures in a consistent and simpler format using this result and the generalized framework in Section~\ref{Sec:Genralized}.
    
    \subsection{On the Maximal Leakage}\label{u:is:redundantA}

        Upon a closer look, we realize that maximal leakage in~\ref{eq:maxLeak} has been defined for a special gain $g_{id}:\W \times \U \rightarrow \Real$ as:
        \begin{align*}
            g_{\text{id}}(w,u)=\begin{cases}
                1, & w=u,\\
                0, & w\neq u,
            \end{cases}
        \end{align*}
        that leads to the following vulnerabilities:
        \begin{align*}
            V_{g_{\text{id}}}(p_{U})&=\max_{u}p(u),\\
            \widehat{V}_{g_{\text{id}}}[p_{U},C]&=\sum_{y\in\Y}\max_{u}p(u,y).
        \end{align*}
        Accordingly, maximal leakage is given as:
        \begin{align*}
          \sup_{U-X-Y}\Lk^{\times}_{g_{\text{id}}}(p_{U},C)=\sup_{U-X-Y}\log\frac{\widehat{V}_{g_{\text{id}}}[p_{U},C]}{V_{g_{\text{id}}}(p_{U})}.
        \end{align*} 
       Therefore, it is natural to extend maximal leakage by incorporating $f, h$ functions, as well as a general gain function $g$, into it as follows.  We call this \textit{\textbf{generalized maximal leakage}}:
        \begin{align}
            \sup_{U-X-Y}\Lk^{\times}_{h,f,g}\left(p_{U},C\right)= \sup_{U-X-Y} \log \frac{\widehat{V}_{h,f,g}[p_{U},C]}{V_{f,g}(p_{U})}.
        \end{align}
        The generalized maximal leakage includes maximal leakage as a special case when \(f\) and $h$ are affine and $g= 
        g_{id}$.
        The main question is then which parts of the above formulation are essential to obtaining this generalized maximal leakage and which are superfluous and can, hence, be dropped without affecting the generality of results.
        In the following, we prove that the generalized maximal leakage for any gain function $g:\W \times \U \rightarrow \Real$ and for given $f,h$ functions is equivalent to the generalized $g$-leakage capacity for the same gain function over the alphabet $\W \times \X$ when we take the supremum over all priors in $ \mathbb{D}\mathcal{\X}$. That is, the introduction of $U$ is superfluous. Similar findings have been reported in~\cite{2023PML,2022EXplainEps}.
        \begin{theorem}\label{thm:maximal=capcity}
            For fixed  \(f\) and $h$, the generalized maximal leakage for a given gain function $g: \W \times \U \rightarrow \Real$  is equivalent to the generalized multiplicative leakage capacity of the same gain function $g: \W \times \X \rightarrow \Real$. That is:
            \begin{align*}
                \sup_{U-X-Y}\Lk^{\times}_{h,f,g}(p_{U},C)=\sup_{\pi}\mathcal{L}_{h,f,g}^{\times}(\pi,C)=\mathcal{L}_{h,f,g}^{\times}(\forall,C).
            \end{align*}
            The proof is given in  Appendix~\ref{appen:maximal=capcity}.
        \end{theorem}
       
        \subsection{Interpretation of \texorpdfstring{$\alpha$}{alpha}-based Leakage Measures}\label{sec:interpret}
      
            In~\cite{2014AdditiveMultpliQIF}, a special gain function was proposed where $\W$ is the set of all probability distributions $w$ on $\X$, $w:\mathbb{D}\X$, and $ g(w,x)=\log w_x,$ where $w_x \in [0,1]$ and $\sum_{x\in\X}w_x=1$. 
            We use the exponential form of this gain function:
            \begin{align}\label{eq:g(w,x)=wx}
                g(w,x)=w_x.
            \end{align}
            For $\alpha \in [0,\infty]$, the \(f\)-mean function $f_{\alpha}:\!\Real^{+}\!\!\!\rightarrow\!\Real$ and its inverse are: 
            \begin{align}\label{eq:alpha-f}
                f_{\alpha}(t)=t^{\frac{\alpha-1}{\alpha}},~ f_{\alpha}^{-1}(s)=s^{\frac{\alpha}{\alpha-1}}.
            \end{align}
         
            \begin{theorem}\label{theo:alpha-vunle}
                For $g(w,x)$ and $f_{\alpha}$ in~\eqref{eq:g(w,x)=wx} and~\eqref{eq:alpha-f} and $h=f_{\alpha}$:
                \begin{align}
                    V_{f_{\alpha},g}(\pi)&=\exp{\Big(-H_{\alpha}(\pi)\Big)},~~ &\alpha \in[0,\infty], \label{eq:alphavulprior}\\ 
                    \widehat{V}_{f_{\alpha},f_{\alpha},g}[\pi,C]&=\exp\Big(-H_{\alpha}(X|Y)\Big), ~~ &\alpha \in[0,\infty] \label{eq:alphavulposter}.
               \end{align}
            \end{theorem}
        
            \begin{proof}
                For $\alpha \in [0,\infty]$, $f_{\alpha}^{-1}$ is convex. Hence, it is valid to be used in the generalized prior vulnerability (\ref{eq:G(PX)}). Thus:
                \begin{align}
                    &V_{f_{\alpha},g}(\pi)=\sup_{w\in\W}f_{\alpha}^{-1}\bigg(\sum_{x\in\X}\pi_{x}f_{\alpha}\left(g(w,x)\right) \bigg) \nonumber\\ 
                    &=\begin{cases}\label{eq:ReyniEnt-inf-sup}
                    \displaystyle \bigg( \inf_{w\in\W} \sum_{x\in\X} \pi_{x} (w_{x})^{\frac{\alpha-1}{\alpha}}\bigg)^{\frac{\alpha}{\alpha-1} } & \alpha \in [0,1),\\ 
                    \displaystyle \bigg( \sup_{w\in\W} \sum_{x\in\X}  \pi_{x}\left(w_{x}\right)^{\frac{\alpha-1}{\alpha}} \bigg)^{\frac{\alpha}{\alpha-1} } & \alpha \in [1,\infty].
                    \end{cases}
                \end{align}
                For $\!\alpha\!\in\![0,1)$, $f_{\alpha}^{-1}$ is decreasing and $\sup_{w\in\W}$ becomes $\inf_{w\in\W}$ when moved inside $f_{\alpha}^{-1}$. With $f_{\alpha}$ being convex in this range, the optimization is also convex with the solution~\cite{2019TunMsurInfLeak_PUT}:
                \begin{align}\label{eq:w* Vfg}
                    w_{x}^{*}=\frac{\pi_{x}^{\alpha}}{\sum_{x\in\X}\pi_{x}^{\alpha}}.
                \end{align}
                For $\alpha \in [1,\infty]$, $f_{\alpha}^{-1}$ is increasing and $f_{\alpha}$ is concave. Thus, the optimization is still convex with the same solution in~\eqref{eq:w* Vfg}.
                Applying the optimal answer $w_{x}^{*}$ in~\eqref{eq:ReyniEnt-inf-sup}, we have:
                \begin{align*}
                    V_{f_{\alpha},g}(\pi)=\bigg(\sum_{x\in\X}\pi_{x}^{\alpha}\bigg)^{\frac{1}{\alpha-1}}=\exp\Big(-H_{\alpha}(\pi)\Big).
                \end{align*} 
                
                For $\widehat{V}_{h,f_{\alpha},g}[\pi,C]$, let $h=f_{\alpha}$ in~\eqref{eq:h(t)posterior} to obtain
                \begin{align*}
                    &\widehat{V}_{f_{\alpha},f_{\alpha},g}[\pi,C] 
                    =f_{\alpha}^{-1}\bigg( \sum_{y\in\Y}p(y)f_{\alpha}\big(V_{f_{\alpha},g}(\delta_{y})\big) \bigg)\nonumber \\
                    &=\bigg( \sum_{y\in\Y}p(y)\Big(\sum_{x\in\X}(\delta_{x}^{y})^{\alpha}\Big)^{\frac{1}{\alpha}}\bigg)^{\frac{\alpha}{\alpha-1}}
                    =\exp\big(-H_{\alpha}(X|Y)\big). \hspace{0.2cm}\qedhere
                \end{align*}  
            \end{proof}       
            \begin{proposition}\label{prop:Vfgconvex}
                $V_{f_{\alpha},g}(\pi)$  in~\eqref{eq:alphavulprior} is convex for $\alpha \in [0,\infty]$. 
            \end{proposition}
            \begin{proof}
                We can write $V_{f_{\alpha},g}(\pi)$ as:
                \begin{align*}
                    V_{f_{\alpha},g}(\pi)=\Big(||\pi||_{\alpha}\Big)^{\frac{\alpha}{\alpha-1}}=f^{-1}_{\alpha}\Big(||\pi||_{\alpha}\Big).
                \end{align*}
                The function $f_{\alpha}^{-1}(s)=s^{\frac{\alpha}{\alpha-1}}$ is convex in the whole range of \(\alpha\) and norm $||\pi||_{\alpha}$ is convex and non-decreasing for $\alpha \in [1,\infty]$ and concave and non-increasing for $\alpha \in [0,1)$. Thus, their composition is convex for $\alpha \in [0,\infty]$. This proposition indicates that $ V_{f_{\alpha},g}$ satisfies axioms of prior vulnerability.
            \end{proof}
     
            \begin{remark}\label{remark:Niloss}
                In~\cite{ding2024cross} an expected loss function $\exp(H_{\alpha}(\pi))$ has been introduced in a similar vein as Thm.~\ref{theo:alpha-vunle}. As an expected loss, it can be a candidate for uncertainty measures. However, it is not concave for $0.5 \leq \alpha$, which contradicts the axiom of prior uncertainty. 
            \end{remark}

    \begin{proposition}[\textbf{$\boldsymbol{\alpha}$-leakage}]\label{prop:Arimoto leakage}
        Using the multiplicative definition of leakage for $V_{f_{\alpha},g}(\pi)$ and  $\widehat{V}_{f_{\alpha},f_{\alpha},g}[\pi,C]$ in~\eqref{eq:alphavulprior} and~\eqref{eq:alphavulposter} we obtain \(\alpha\)-leakage as
        \begin{align}
            \Lk^{\times}_{f_{\alpha},f_{\alpha},g}(\pi,C)&=\log\frac{\widehat{V}_{f_{\alpha},f_{\alpha},g}[\pi,C]}{V_{f_{\alpha},g}(\pi)} =\log\frac{\exp\left(-H_{\alpha}(X|Y)\right)}{\exp\left(-H_{\alpha}(\pi)\right)} \nonumber\\
                       &=H_{\alpha}(X)-H_{\alpha}(X|Y)=I_{\alpha}^{A}(X,Y).\label{eq:IAalpha}
        \end{align}
    \end{proposition}
    \begin{remark}\label{remark:maxalphaleak}
        Equation~\eqref{eq:IAalpha} represents \(\alpha\)-leakage in~\cite{2019TunMsurInfLeak_PUT} consistently within our generalized framework. 
        Moreover, the \textbf{maximal $\boldsymbol{\alpha}$-leakage} is given by Thm.~\ref{thm:maximal=capcity}:
        \begin{align}
            \Lk_{f_{\alpha},f_{\alpha},g}^{\times}(\forall,C)=\sup_{\pi}I_{\alpha}^{A}(X,Y), \quad \alpha \in [0,\infty].
        \end{align}
        This interpretation extends \(\alpha\) to the whole $[0,\infty]$ range originally defined by Arimoto in~\cite{arimoto1977information}. Note that for $\alpha=\infty$ we have $\Lk^{\times}_{f_{\infty},f_{\infty},g}(\forall,C)=\mathcal{ML}(C).$
    \end{remark}
     Next, we show that the generalized framework can express maximal $(\alpha,\beta)$-leakage. 
    \begin{proposition}[\textbf{Maximal} $(\boldsymbol{\alpha,\beta})$-\textbf{leakage}~\cite{2022Alphabetleakage}]\label{prop:maximalalphabeta}
        Consider the same $g$ and $f_{\alpha}$ given in~\eqref{eq:g(w,x)=wx} and~\eqref{eq:alpha-f}. Let $h_{(\alpha,\beta)}:\Real^{+}\rightarrow\Real$ be $h_{(\alpha,\beta)}(t)=t^{\frac{(\alpha-1)\beta}{\alpha}}$ and $h^{-1}_{(\alpha,\beta)}(s)=s^{\frac{\alpha}{(\alpha-1)\beta}}$, where $\alpha \in (1,\infty]$ and $\beta \in [1,\infty]$. The generalized leakage is given by:
        {\interdisplaylinepenalty=10000
        \begin{align}
            &\Lk^{\times}_{h_{(\alpha,\beta)},f_{\alpha},g}\left(\pi,C\right)\nonumber\\
            &=\frac{\alpha}{(\alpha-1)\beta}\log \sum_{y\in\Y}p(y)^{1-\beta} \left[ \frac{\sum_{x\in\X}\pi_{x}^{\alpha}C_{x,y}^{\alpha}}{{\sum_{x\in\X}\pi_{x}^{\alpha}}} \right]^{\frac{\beta}{\alpha}}.\label{eq:alphabeatleak}
        \end{align}
        }
        
        Then the maximal $(\alpha,\beta)$-leakage is given by Thm.\ref{thm:maximal=capcity} and the generalized capacity as:
        \begin{align}\label{eq:alphabetacap}
            \Lk^{\times}_{h_{(\alpha,\beta)},f_{\alpha},g}(\forall,C)=\sup_{\pi}\Lk^{\times}_{h_{(\alpha,\beta)},f_{\alpha},g}(\pi,C).
        \end{align}
        The proof is provided in Appendix~\ref{app:proofofalphabeta}.
    \end{proposition}
        \begin{remark}
            To keep it consistent within our framework, we need to check the ranges of \(\alpha\) and $\beta$ such that $h$ is convex and increasing or concave and decreasing. We have :
            \begin{align*}
                &h'_{(\alpha,\beta)}(t)=\frac{(\alpha-1)\beta}{\alpha}t^{\frac{(\alpha-1)\beta}{\alpha}-1},\\
                &h''_{(\alpha,\beta)}(t)=\left(\frac{(\alpha-1)\beta}{\alpha}\right)\left(\frac{(\alpha-1)\beta}{\alpha}-1\right)t^{\frac{(\alpha-1)\beta}{\alpha}-2}.
            \end{align*}
            For $\alpha \in (1,\infty]$ and $\beta \in [1,\infty]$, $h'_{\alpha,\beta}(t)\geq 0$ and it is increasing, thus $h_{(\alpha,\beta)}$ should be convex. If $\beta \geq \frac{\alpha}{\alpha-1}$ then $h''_{(\alpha,\beta)}(t)\geq 0$ and the function is convex. 
        \end{remark}
        \begin{remark}\label{remark:specialcasealphabeta}
        
            We demonstrate that our generalized result can achieve all special cases of maximal $(\alpha,\beta)$-leakage as presented in~\cite{2022Alphabetleakage}.
            \begin{enumerate}
                \item \textbf{Maximal $\!\alpha$-leakage} ($\beta=1$): If $\beta=1$ then $h_{(\alpha,1)}=f_{\alpha}$ and this case is given by Proposition~\ref{prop:Arimoto leakage} and Remark~\ref{remark:maxalphaleak}.
                
                \item \textbf{Maximal leakage} ($\alpha=\infty$, $\beta=1$): This case is easily given by the maximal \(\alpha\)-leakage when $\alpha=\infty$.
                
                \item \textbf{R{\'e}nyi LDP} ($\alpha=\beta$): It is given by~\eqref{eq:alphabeatleak} and~\eqref{eq:alphabetacap} as:
                \begin{align}\label{eq:localreyni}
                    &\Lk^{\times}_{h_{(\alpha,\alpha)},f_{\alpha},g}(\forall,C)=\max_{x,x'}\frac{1}{\alpha-1}\log\sum_{y\in\Y}C_{x',y}^{1-\alpha}C_{x,y}^{\alpha}.
                \end{align}
                The proof of this item is given in Appendix~\ref{app:proofoflocalreyni}.
                
                \item \textbf{LDP} ($\alpha\!=\!\beta\!=\!\infty$): It is given by R{\'e}nyi LDP for $\alpha=\infty$:
                \begin{align*}
                     \Lk^{\times}_{h_{(\infty,\infty)},f_{\infty},g}(\forall,C)
                     =\log\max_{y\in\Y}\frac{\displaystyle\max_{x\in\X}C_{x,y}}{\displaystyle \min_{x\in\X}C_{x,y}}=\Lk^{\text{LDP}}(C).
                \end{align*}
                 
                 \item \textbf{Differential Privacy:} In~\cite{2024GilaniUnifyingMabel}, it was demonstrated that with vectorized inputs and a conditional definition of maximal $(\alpha,\beta)$-leakage, R{\'e}nyi differential privacy and standard differential privacy are achievable when $\alpha = \beta$ and $\alpha = \beta = \infty$. We omit the details for brevity.
            \end{enumerate}
        \end{remark}

        It is notable that when $\alpha = \beta = \infty$, maximal \(\alpha\)-leakage and maximal $(\alpha,\beta)$-leakage reduce to two primary capacity measures: Bayes capacity and LDP, respectively. These represent the worst-case leakage measures known so far, where the former is the average-case and the latter the max-case capacity~\cite{2022EXplainEps}. However, these results have been derived when the adversary applies normal averaging, rather than generalized \(f\)-mean averaging. The key question is:  what are the worst-case capacities, beyond these special cases, for arbitrary functions \(f\) and $h$? The following section provides partial answers to this question for the class of multiplicative $f^{-1}$. Incidentally, the $f^{-1}$ functions we had to apply to obtain maximal \(\alpha\) leakage and maximal $(\alpha, \beta)$ leakage in the generalized QIF framework are multiplicative. Establishing a comprehensive and general capacity result remains an open problem.

\section{Some Results on The Generalized Leakage and Capacity}\label{sec:capacity}    

    In this section, we first derive the multiplicative generalized leakage capacity in the special case that $f = h$ and $f^{-1}$ is a multiplicative function (that is, in the special case when the adversary uses the same averaging function for computing its best posterior and prior actions and  $f^{-1}(ab) = f^{-1}(a)f^{-1}(b)$ for all $a,b$ in the domain of $f^{-1}$). The multiplicative generalized leakage capacity turns out to be the $\log f^{-1}\left(\exp{\left(\mathcal{ML}(C)\right)}\right)$---which can exceed the exponent of the Bayes capacity \(\exp{\big(\mathcal{ML}(C)\big)} \) when \( f^{-1}(x) > x \). 
    We also show that the max-case leakage capacity is the $ \log f^{-1}\left( \exp{\left(\Lk^{\text{LDP}}(C)\right)}\right)$ when $f^{-1}$ is increasing and $ \log f^{-1}\left( \exp{\left(-\Lk^{\text{LDP}}(C)\right)}\right)$ when $f^{-1}$ is decreasing. 
    We then establish the fact that the generalized average posterior vulnerability is upper bounded by the generalized maximum posterior vulnerability (for the same $f,h$). We then use this fact to report partial results on generalized leakage and capacity measures when $f \neq h$, but $f^{-1}$ is still a multiplicative function. 
       
    \begin{theorem}
        For any valid $h,f,g, \pi$ and $C$, if $h=f$ and $f^{-1}$ is multiplicative, i.e., $f^{-1}(ab)$ $= f^{-1}(a)f^{-1}(b),$ for all  $a,b$ in the domain of $f^{-1}$, then
        \begin{align}
            \mathcal{L}^{\times}_{f,f,g}(\pi,C) \leq \log f^{-1}\bigg(\sum_{y\in\Y} \max_{x\in\X} C_{x,y} \bigg).
        \end{align}
    \end{theorem}
    \begin{proof} First we show the result for $\widehat{V}_{f,f,g}[\pi,C]$ and increasing \(f\) and $f^{-1}$:
        \begin{align}
            &\widehat{V}_{f,f,g}[\pi,C]=f^{-1}\!\bigg(\sum_{y\in\Y} \sup_{w\in\W} \sum_{x\in\X} C_{x,y}\pi_{x} f\big( g(w, x) \big)\!\!\bigg) \label{eq:vhath=f}\\
            &\leq f^{-1}\bigg(\sum_{y\in\Y} \sup_{w\in\W} \sum_{x\in\X} \left(\max_{x\in\X} C_{x,y}\right) \pi_{x} f\big( g(w, x) \big) \bigg) \\
            &= f^{-1}\bigg(\sum_{y\in\Y} \max_{x\in\X} C_{x,y} \sup_{w\in\W} \sum_{x\in\X} \pi_{x} f\big( g(w, x) \big) \bigg)\\
            &=f^{-1}\bigg(\sum_{y\in\Y} \max_{x\in\X} C_{x,y} \!\bigg)\!f^{-1}\!\bigg(\!\sup_{w\in\W}\!\sum_{x\in\X} \!\pi_{x} f\big( g(w, x)\big) \!\!\bigg) \label{eq:multiplf-1} \\
            &=f^{-1}\bigg(\sum_{y\in\Y} \max_{x\in\X} C_{x,y}\! \bigg)V_{f,g}(\pi),
        \end{align}
        where~\eqref{eq:vhath=f}  is given by~\eqref{eq:h=f} for increasing \(f\) and~\eqref{eq:multiplf-1} is due to the assumption about $f^{-1}$ being multiplicative. The same relations are held for decreasing \(f\) and $f^{-1}$.
        Now, for the generalized multiplicative leakage, we have:
        \begin{align*}
            \mathcal{L}^{\times}_{f,f,g}(\pi,C)
            & = \log\frac{\widehat{V}_{f,f,g}[\pi,C]}{V_{f,g}(\pi)}\\
            & \leq \log f^{-1}\bigg(\sum_{y\in\Y} \max_{x\in\X} C_{x,y} \bigg). \hspace{2.2cm}\qedhere
        \end{align*}
    \end{proof}
    
    It can be shown that the equality above is achieved when we select a uniform prior $\pi$, and we have $f\big(g(w,x)\big)=ag_{\text{id}}$ for any $a>0$:
    \begin{align}\label{eq:fgwxid}
        f\big(g(w,x)\big)=\begin{cases}
            a, &w=x,\\
            0, &\text{otherwise}.
        \end{cases}
    \end{align}
     Under the special case where $h=f=f^{-1}$ is the identity function, we will recover the Bayes capacity. We capture these results in the following theorem.
    \begin{theorem}
        For any valid $h,f,g, \pi$ and $C$, if $h=f$ and $f^{-1}$ is multiplicative, then
        \begin{align}\label{eq:f-1MAXL}
            \mathcal{L}^{\times}_{f,f,\forall}(\forall,C) = \log f^{-1}\bigg(\sum_{y\in\Y} \max_{x\in\X} C_{x,y} \bigg).
        \end{align}     
    \end{theorem}
    \begin{proof}
         Let $\pi_{x}=\frac{1}{|\X|}$ and $f\big(g(w,x)\big)$ be the function in~\eqref{eq:fgwxid} for an increasing \(f\). Then,
        \begin{align*}
            V_{f,g}(\pi)&=f^{-1}\bigg(\sup_{w\in\W}\sum_{x\in\X} \pi_{x}f\big( g(w, x) \big) \bigg)
            =f^{-1}\left(\frac{a}{|\X|} \right),
        \end{align*}
        where $ a=\sup_{w\in\W}\sum_{x\in\X}f\big(g(w,x)\big)$. 
        Similarly for ${V}_{f,f,g}\Hyper{\pi}{C}$:
        \begin{align*}
            \widehat{V}_{f,f,g}[\pi,C]&=f^{-1}\bigg(\sum_{y\in\Y} \frac{1}{|\X|}  \sup_{w\in\W} \sum_{x\in\X} C_{x,y}f\big( g(w, x) \big) \bigg)\\
            &=f^{-1}\bigg(\frac{a}{|\X|}\sum_{y\in\Y} \max_{x\in\X}C_{x,y} \bigg)\\
            &=f^{-1}\bigg(\frac{a}{|\X|} \bigg)f^{-1}\bigg(\sum_{y\in\Y} \max_{x\in\X}C_{x,y} \bigg)\Rightarrow\\
            &\mathcal{L}^{\times}_{f,f,g}(\pi,C)= \log f^{-1}\bigg(\sum_{y\in\Y} \max_{x\in\X} C_{x,y} \bigg).~~~~~\qedhere
        \end{align*}
    \end{proof}
        
    \begin{theorem}\label{thm:LDPf-1}
        For any valid $f,g, \pi$ and $C$, if $f^{-1}$ is multiplicative, i.e., $f^{-1}(ab)$ $=f^{-1}(a)f^{-1}(b),$ for all $a,b$ in the domain of $f^{-1}$, we have:
        \begin{align*}
            &\mathcal{L}^{\max}_{f,\forall}(\forall,C)=\sup_{\pi,g}\mathcal{L}^{\max}_{f,g}(\pi,C)  \\
            &\leq
            \begin{cases}
                 \displaystyle \log \max_{y\in\Y} f^{-1}\left( \frac{\max_{x\in\X}C_{x,y}}{\min_{x\in\X}C_{x,y}}\right), & f^{-1} \text{increasing}, \\
                  \displaystyle \log  \max_{y\in\Y} f^{-1}\left( \frac{\min_{x\in\X}C_{x,y}}{\max_{x\in\X}C_{x,y}}\right). & f^{-1} \text{decreasing.}
            \end{cases}
        \end{align*}
          For an increasing $f^{-1}$, we can transfer $\max_{y\in\Y}$ into the function resulting in $ f^{-1}\left(\Lk^{\text{LDP}}(C)\right)$. When $f^{-1}$ is increasing, $\max_{y\in\Y}$ becomes $\min_{y\in\Y}$ when it is moved inside the function and we have  $ f^{-1}\left(\left(\Lk^{\text{LDP}}(C)\right)^{-1}\right)$.\\
          See Appendix~\ref{app:proofofLDPf-1} for the proof.
    \end{theorem}

   The following result establishes the fact that the generalized average posterior vulnerability is upper bounded by the maximum generalized posterior vulnerability. We will then use this result in the next proposition. 
    \begin{lemma} \label{lemm:vhg<vmax} For any valid $h,f,g,\pi$, and  $C$ we have:
        \begin{align}
             \widehat{V}_{h,f,g}[\pi,C] \leq \widehat{V}^{\max}_{f,g}[\pi,C].
        \end{align}
    \end{lemma}
    \begin{proof}
        For decreasing $h$ and $h^{-1}$ we have:
        \begin{align}
            \widehat{V}_{h,f,g}[\pi,C]
            &=h^{-1}\bigg( \sum_{y\in\Y}p(y) h\Big(V_{f,g}(\delta^{y})\Big) \bigg)\\
            &\leq h^{-1}\bigg(\min_{y\in\Y} h\Big(V_{f,g}(\delta^{y})\Big) \bigg) \label{eq:h-1max}\\
            &= h^{-1}\bigg( h\Big(\max_{y\in\Y}V_{f,g}(\delta^{y})\Big) \bigg) \label{eq:h-min-max}\\
            &=\max_{y\in\Y}V_{f,g}(\delta^{y})=\widehat{V}^{\max}_{f,g}[\pi,C].
        \end{align}
       In~\eqref{eq:h-min-max} $\min_{y\in\Y}$ turns into $\max_{y\in\Y}$ since $h$ is decreasing. 
       The proof is similar for an increasing $h$, but in~\eqref{eq:h-1max} we use $\max_{y\in\Y}$.
    \end{proof}
   
    \begin{proposition}\label{Thm.GenLDP}
        For any valid $h,f,g, \pi$ and $C$, if $h\neq f$ and $f^{-1}$ is multiplicative, i.e., $f^{-1}(ab)$ $=           f^{-1}(a)f^{-1}(b),$ for all  $a,b$ in the domain of $f^{-1}$, we have:
        \begin{align}\label{eq:LDP-1f}
            &\mathcal{L}^{\times}_{h,f,\forall}(\forall,C)=\sup_{\pi,g}\mathcal{L}^{\times}_{h,f,g}(\pi,C) \leq \mathcal{L}^{\max}_{f,\forall}(\forall,C)\\
            &\leq
            \begin{cases}
                 \displaystyle\log \max_{y\in\Y} f^{-1}\left( \frac{\max_{x\in\X}C_{x,y}}{\min_{x\in\X}C_{x,y}}\right), & f^{-1} \text{ increasing},\nonumber\\
                 \displaystyle \log \max_{y\in\Y} f^{-1}\left( \frac{\min_{x\in\X}C_{x,y}}{\max_{x\in\X}C_{x,y}}\right). & f^{-1} \text{ decreasing},
            \end{cases}
        \end{align}
        where the inequalities are given by Lemma~\ref{lemm:vhg<vmax} and Theorem~\ref{thm:LDPf-1}.
    \end{proposition}
    
    Bayes capacity and LDP leakage represent the worst-case scenarios for average-case and max-case leakage, respectively. The results in this section provide an upper bound on these capacities, potentially introducing new worst-case scenarios. However, since both results in~\eqref{eq:f-1MAXL} and~\eqref{eq:LDP-1f} are expressed as functions of Bayes capacity and LDP leakage, defining a new worst-case scenario requires specific conditions—particularly when \(f^{-1}(x)>x\). The existence of an operationally meaningful function \( f^{-1} \) that satisfies this condition remains an open research question.

\section{Pointwise \texorpdfstring{$\alpha$}{alpha}-leakage and Sibson mutual information}\label{Sec:Sibson}

    In addition to the \(\alpha\)-based measures we considered in Section~\ref{u:is:redundant}, two other important measures with applications in privacy are R{\'e}nyi divergence and Sibson mutual information. 
    To our knowledge, their operational meaning cannot be expressed via the proposed generalized leakage measures as a ratio of separate posterior and prior vulnerabilities. 
    In this section, we are inspired by R{\'e}nyi's interpretation of information gain to propose an operational meaning for R{\'e}nyi divergence and Sibson mutual information for the whole range of $\alpha \in [0,\infty]$  as a measure of privacy.  \footnote{We note that~\cite{ding2024alpha} made some progress in this direction and showed how R\'enyi divergence and Sibson mutual information can be interpreted as \(f\)-mean information gain measures. However, the current paper takes a more general approach and shows how these measures fit within the proposed QIF generalized framework.}.

    The definition of Bayes and  $g$-vulnerability in~\cite{2014AdditiveMultpliQIF} are rooted in the concept of entropy and uncertainty in information theory. 
    In~\cite{1961measuresreny}, R{\'e}nyi generalized Shannon entropy by relaxation of one of the five postulates considered for a measure of uncertainty to achieve new definitions of entropy. 
    Moreover, he also generalized the characterization of the amount of information which led to R{\'e}nyi divergence~\cite[Sec. 3]{1961measuresreny}. 
    In his seminal paper,~\cite[Sec. 3]{1961measuresreny},  R{\'e}nyi provides elegant remarks, which we quote: 
    ``\textit{Entropy can be interpreted not only as a measure of uncertainty but also as a measure of information}.''
    
    Then he continued by saying there are other ways to quantify the amount of information. 
    ``\textit{For instance, we may ask what is the amount of information concerning a random variable $\zeta$ obtained from observing an event $E$, which is in some way connected with the random variable $\zeta$. If P denotes the original (unconditional) distribution of the random variable $\zeta$ and $Q$ the conditional distribution of $\zeta$ under the condition that the event $E$ has taken place, we shall denote a measure of the amount of information concerning the random variable $\zeta$ contained in the observation of the event $E$ by $I(Q|P)$.''} (In most subsequent literature, $I(Q|P)$ has been changed to $D(P|Q)$.)
   
    By considering five postulates for the amount of information~\cite{1961measuresreny}, R{\'e}nyi proved that the amount of information obtained about each $x\!\in\!\X$ by a singleton observation $y\!\in\!\Y$  is given by \(\displaystyle \log\frac{\delta^{y}_{x}}{\pi_{x}}\). 
    Accordingly, we propose a new gain function as an information gain, leading to a new leakage measure called \textbf{\textit{pointwise information gain}}. By this definition, we quantify the information gain of each channel's output $y \in \Y$ according to its corresponding inner $\delta^{y}$. Interestingly, this definition of leakage satisfies axioms of information measure but is not given by separate prior and posterior vulnerabilities. 
    
    \begin{definition}\label{def:information gain and leakage}
        Let $\W$ be the (uncountable infinite) set of all probability distributions $w$ on $\X$. The pointwise information gain is defined as:
        \begin{align}
            \gamma(w,x)=\log\frac{w_{x}}{\pi_{x}}.
        \end{align}
            Then, for a given strictly monotonic and continuous function  $\ell$, the\textit{ \textbf{generalized pointwise posterior leakage}} is defined as:
        \begin{align}\label{eq:informative leakage}
            \mathcal{I}_{\ell,\gamma}(\pi,\delta^{y}) =\sup_{w\in\W}\ell^{-1}\bigg(\sum_{x\in\X}\delta_{x}^{y}\ell\big(\gamma(w,x)\big) \bigg).
        \end{align}
    \end{definition}
    In~\eqref{eq:informative leakage}, we take the expectation w.r.t to $\delta^{y}$ since $w$ quantifies the information gain for each $x$ after observation of $y \in \Y$. Consequently, the \textbf{\textit{generalized average posterior leakage}} is:
    \begin{align}
        \mathcal{I}_{h,\ell,\gamma}(\pi,C)=h^{-1}\bigg(\sum_{y\in\Y}p(y)h\Big(\mathcal{I}_{\ell,\gamma}(\pi,\delta^{y})\Big) \bigg),
    \end{align}
    where $h$ is a strictly monotonic and continuous function.
    
    Since  $\mathcal{I}_{\ell,\gamma}(\pi,\delta^{y})$ provides information leakage for each $y \in \Y$, it is a proper candidate for the operational meaning of a class of privacy measures called pointwise measures. An example of pointwise measures is pointwise maximal leakage~\cite{2023PML}, which is an extension of maximal leakage that quantifies leakage for each $y \in \Y$ and is given by:
    \begin{align*}
        \mathfrak{L}(X\rightarrow y)=\log\max_{x \in \suport{\pi}}\frac{\delta_{x}^{y}}{\pi_{x}}=D_{\infty}(\delta_{y}\|\pi)=D_{\infty}(p_{X|y}\|\pi),
    \end{align*}
    where $D_{\infty}(\delta_{y}\|\pi)$ is R{\'e}yni divergence of order $\infty$.
    
    In~\cite{2023PML}, pointwise maximal leakage has been proposed in the $U$ function framework. However, the authors also proved that this framework is equivalent to the $g$-leakage framework.
    Here, we use generalized pointwise posterior leakage to propose the operational meaning of R{\'e}yni divergence and Sibson mutual information for a whole range of $\alpha \in [0,\infty]$.
    Then, pointwise maximal leakage is given as the special case of $\alpha=\infty$.
    
    \begin{definition}[\textbf{Pointwise \(\alpha\)-leakage}] \label{def:P-alpha-leak}
        Let $\ell_{\alpha}(t)=\exp(\frac{\alpha-1}{\alpha}t)$ with inverse $\ell_{\alpha}^{-1}(s)=\frac{\alpha}{\alpha-1}\log(s)$.
        For each $y \in \Y$ with inner $\delta^{y}$ given by the hyper $\Delta=[\pi,C]$, pointwise \(\alpha\)-leakage is defined as:
        \begin{align}\label{eq:P-alpha-leak} \nonumber      
            &\mathcal{I}_{\ell_{\alpha},\gamma}(\pi,\delta^{y}) \triangleq 
            \sup_{w\in\W} {\ell_{\alpha}}^{-1} \bigg( \sum_{x\in\X} \delta_{x}^{y} \ell_{\alpha}\big(\gamma(w,x)\big)\bigg)\\
            =&\begin{cases}
                \displaystyle \frac{\alpha}{\alpha-1}\log  \inf_{w\in\W} \sum_{x\in\X} \delta_{x}^{y} \left(\frac{w_{x}}{\pi_{x}}\right)^{\frac{\alpha-1}{\alpha}}, & \alpha \in [0,1),\\ 
                \displaystyle \frac{\alpha}{\alpha-1}\log  \sup_{w\in\W} \sum_{x\in\X}  \delta_{x}^{y}\left(\frac{w_{x}}{\pi_{x}}\right)^{\frac{\alpha-1}{\alpha}}, & \alpha \in [1,\infty].
            \end{cases}
        \end{align}
    \end{definition}

    \begin{theorem} 
        Pointwise \(\alpha\)-leakage is the R{\'e}nyi divergence of order $\alpha \in [0,\infty]$ between $\delta^{y}$ and $\pi$:
        \begin{align*}\label{eq:post=reyni}
             \mathcal{I}_{\ell_{\alpha},\gamma}(\pi,\delta^{y})&=\frac{1}{\alpha -1}\log\sum_{x\in\X}\left(\delta_{x}^{y}\right)^{\alpha}\pi_{x}^{\alpha-1}=D_{\alpha}\left(\delta^{y}\|\pi\right).
        \end{align*}
     \end{theorem}
     \begin{proof}
        In a similar vein to the proof of Theorem~\ref{theo:alpha-vunle}, both optimizations in~\ref{eq:P-alpha-leak} are convex and have the solution:
         \begin{align*}
             w^{*}_{x}=\frac{ (\delta_{x}^{y})^{\alpha}/(\pi_{x})^{\alpha-1} }{ \sum_{x\in\X}(\delta_{x}^{y})^{\alpha}/(\pi_{x})^{\alpha-1} }.
         \end{align*}
         By replacing $w^{*}_{x}$ in~\eqref{eq:P-alpha-leak}, the result is achieved.
     \end{proof}

    \begin{proposition}
        For $h=\ell_{\alpha}$, the generalized average of pointwise \(\alpha\)-leakage is Sibson mutual information:  
    \end{proposition}
    \begin{align}
        \mathcal{I}_{\ell_{\alpha},\ell_{\alpha},\gamma}(\pi,C)=I_{\alpha}^{S}(X;Y).
    \end{align}
    \begin{proof}
        We expand  $\mathcal{I}_{\ell_{\alpha},\ell_{\alpha},\gamma}(\pi,C)$ as
        \begin{align*}
            \mathcal{I}_{\ell_{\alpha},\ell_{\alpha},\gamma}(\pi,C)
            &=\ell_{\alpha}^{-1}\bigg(\sum_{y\in\Y}p(y)\ell_{\alpha}\Big(\mathcal{I}_{\ell_{\alpha},\gamma}\big(\pi,\delta^{y}\big)\Big) \bigg)\\
            &=\frac{\alpha}{\alpha-1}\log\Bigg(\sum_{y\in\Y}p(y)\bigg(\sum_{x\in\X}(\delta_{x}^{y})^{\alpha}\pi_{x}^{\alpha-1}\bigg)^{\frac{1}{\alpha}} \Bigg)\\
            &=I_{\alpha}^{S}(X;Y). \hspace{4.55cm} \qedhere
        \end{align*}
    \end{proof}
   
\section{Axiomatic Relations of Generalized Vulnerabilities}\label{sec:axioms}

    At first, we review the axioms of prior and posterior vulnerabilities given in~\cite{2016alvimaxioms}.
    Generic vulnerability measures have been defined as functions of the following types:
    \begin{equation}
        \begin{array}{rll}
            \text{Prior vulnerability}: & \mathbb{V}: \mathbb{D}\mathcal{X} \rightarrow \mathbb{R}^{+}, \nonumber \\
            \text{Posterior vulnerability}: & \widehat{\mathbb{V}}: \mathbb{D}^2\mathcal{X} \rightarrow \mathbb{R}^{+}.
        \end{array}
    \end{equation}
    
    The following axioms are adopted specifically for $\mathbb{V}$:
    \begin{itemize}
        \item \textbf{Continuity} (\CNTY): The vulnerability function \(\mathbb{V}\) is continuous with respect to \(\pi\) (in terms of the standard topology on \( \mathbb{D}\mathcal{X} \)).
        \item \textbf{Convexity}  (\CVX): The vulnerability function \( \mathbb{V} \) is convex in \(\pi\), meaning for all convex combinations 
        \(\sum_{i}a_{i} \pi^{i} \):
        \[\mathbb{V}\Big(\sum_{i} a_{i} \pi^{i}\Big) \leq \sum_{i} a_{i} \mathbb{V}\left(\pi^{i}\right).\]
        \item \textbf{Quasi-convexity} (\QCVX):  \( \mathbb{V} \) is quasi-convex in \( \pi \) where for all convex combinations \( \sum_{i} a_{i} \pi^{i} \):
        \[\mathbb{V}\Big(\sum_{i} a_{i} \pi^{i}\Big) \leq \max_{i} \mathbb{V}\left(\pi^{i}\right).\]
    \end{itemize}
    Based on these axioms, the following results were proven~\cite{2016alvimaxioms}.
    \begin{theorem}[\!\!{\cite[Prop. 2, Thm. 3, Cor. 4]{2016alvimaxioms}}]\label{thm:Vg-satisfy}
        Any \( g \)-vulnerability \( V_g \) satisfies \CNTY~and \CVX.
    \end{theorem}
    \begin{theorem}[\!\!{\!\cite[Thm. 5]{2016alvimaxioms}}]\label{thm:bbV-Vg} 
        Let \(\mathbb{V}:\mathbb{D}\mathcal{X}\rightarrow\mathbb{R}^{+}\) be a vulnerability function satisfying \CNTY~and \CVX. Then there exists a gain function \(g\) with a countable number of guesses such that \(\mathbb{V}\!=\!V_{g}\).
    \end{theorem}
    According to \Thm{thm:Vg-satisfy} and \Thm{thm:bbV-Vg}, without loss of generality, we consider \( V_{g}(\pi) \) as the definition of prior vulnerability.
    
    The following axioms are adopted specifically for posterior vulnerability.
    \begin{itemize}
        \item \textbf{Non-interference} (\NI): The vulnerability of a point-hyper equals the vulnerability of the unique inner distribution of that hyper:
        \[\forall \pi: \quad \widehat{\mathbb{V}}[\pi] = \mathbb{V}(\pi).\]
        \item \textbf{Data-processing inequality} (\DPI): Post-processing does not increase vulnerability:
        \[\forall \pi, C, R: \quad \widehat{\mathbb{V}}[\pi, C] \geq \widehat{\mathbb{V}}[\pi, CR],\]
        where $R$ is any valid channel.
        \item \textbf{Monotonicity} (\MONO): Pushing a prior through a channel does not decrease vulnerability:
        \[\forall \pi, C: \quad \widehat{\mathbb{V}}[\pi, C] \geq \mathbb{V}(\pi).\]
    \end{itemize}
    \begin{figure}[htbp]
        \centering   
        \subfigure[\label{fig:AVGa} Defining $\widehat{\mathbb{V}}$ as \AVG]{\tikzset{every picture/.style={line width=0.75pt}} 

\begin{tikzpicture}[x=0.75pt,y=0.75pt,yscale=-1,xscale=1]

\draw   (284,51) -- (312,51) -- (312,65) -- (284,65) -- cycle ;
\draw   (284,72) -- (312,72) -- (312,86) -- (284,86) -- cycle ;
\draw   (284,101) -- (312,101) -- (312,115) -- (284,115) -- cycle ;
\draw   (324,107) -- (352,107) -- (352,121) -- (324,121) -- cycle ;
\draw  (244,107) -- (272,107) -- (272,121) -- (244,121) -- cycle ;
\draw    (298,86.12) -- (298,98.12) ;
\draw [shift={(298,101.12)}, rotate = 270] [fill={rgb, 255:red, 0; green, 0; blue, 0 } ][line width=0.08]  [draw opacity=0] (3.57,-1.72) -- (0,0) -- (3.57,1.72) -- cycle;

\draw    (312,58) .. controls (321,60) and (338,62) .. (338,103) ;
\draw [shift={(338,106.56)}, rotate = 268.99] [fill={rgb, 255:red, 0; green, 0; blue, 0 }  ][line width=0.08]  [draw opacity=0] (3.57,-1.72) -- (0,0) -- (3.57,1.72) -- cycle;

\draw    (258,107) .. controls (258,62) and (276,60) .. (281,58) ;
\draw [shift={(284.3,58)}, rotate = 173.52] [fill={rgb, 255:red, 0; green, 0; blue, 0 }  ][line width=0.08]  [draw opacity=0] (3.57,-1.72) -- (0,0) -- (3.57,1.72) -- cycle;

\draw    (338,121.88) .. controls (334.62,137.16) and (273.79,138.05) .. (259,123.58) ;
\draw [shift={(258,121.43)}, rotate = 55.78] [fill={rgb, 255:red, 0; green, 0; blue, 0 }  ][line width=0.08]  [draw opacity=0] (3.57,-1.72) -- (0,0) -- (3.57,1.72) -- cycle;

\draw  [color={rgb, 255:red, 0; green, 0; blue, 0 } ,draw opacity=1 ][fill={rgb, 255:red, 255; green, 255; blue, 255 } ,fill opacity=1 ] (274.96,131.24) .. controls (274.97,130.43) and (276.04,129.28) .. (277.45,129.29) .. controls (278.86,129.30) and (279.91,130.43) .. (279.90,131.84) .. controls (279.89,133.25) and (278.82,134.40) .. (277.41,134.39) .. controls (276.00,134.38) and (274.95,133.25) .. (274.96,131.74) -- cycle ;

\draw    (280,131.24) .. controls (278.55,132.88) and (296.8,129.38) .. (297.8,115.38) ;
\draw    (268,72) .. controls (268.8,76) and (278.05,79) .. (284,78.88) ;
\draw    (312,79) .. controls (316,78) and (334,84) .. (337,91) ;

\draw  [color={rgb, 255:red, 0; green, 0; blue, 0 } ,draw opacity=1 ][fill={rgb, 255:red, 255; green, 255; blue, 255 } ,fill opacity=1 ] (334.05,91.34) .. controls (334.06,90.03) and (335.13,88.88) .. (336.54,88.89) .. controls (337.95,88.9) and (339,90.03) .. (338.99,91.44) .. controls (338.98,92.85) and (337.91,94.00) .. (336.50,93.99) .. controls (335.09,93.98) and (334.04,92.85) .. (334.05,91.34) -- cycle ;

\draw  [color={rgb, 255:red, 0; green, 0; blue, 0 } ,draw opacity=1 ][fill={rgb, 255:red, 255; green, 255; blue, 255 } ,fill opacity=1 ] (264.31,69.32) .. controls (264.32,68.51) and (265.39,67.36) .. (266.80,67.37) .. controls (268.21,67.38) and (269.26,68.51) .. (269.25,69.92) .. controls (269.24,71.33) and (268.17,72.48) .. (266.76,72.47) .. controls (265.35,72.46) and (264.30,71.33) .. (264.31,69.82) -- cycle ;


\draw (298.5,58.4) node [anchor=center][inner sep=0.75pt] [align=center] {\fontfamily{ptm}\selectfont {\scriptsize \CVX}};
\draw (298.5,78.25) node [anchor=center][inner sep=0.75pt] [align=center] {\scriptsize {\fontfamily{ptm}\selectfont \AVG}};
\draw (298.5,108) node [anchor=center][inner sep=0.75pt] [align=center] {\fontfamily{ptm}\selectfont {\scriptsize \NI}};
\draw (339,114.5) node [anchor=center][inner sep=0.75pt] [align=center] {\fontfamily{ptm}\selectfont {\scriptsize \DPI}};
\draw (258,114.5) node [anchor=center][inner sep=0.75pt] [align=center] {\fontfamily{ptm}\selectfont {\scriptsize \MONO}};

\end{tikzpicture}}
        \subfigure[\label{fig:AVG} Defining $\widehat{\mathbb{V}}$ as \MAX]{\tikzset{every picture/.style={line width=0.75pt}} 

\begin{tikzpicture}[x=0.75pt,y=0.75pt,yscale=-1,xscale=1]

\draw   (284,51) -- (312,51) -- (312,65) -- (284,65) -- cycle ;
\draw   (284,72) -- (312,72) -- (312,86) -- (284,86) -- cycle ;
\draw   (284,101) -- (312,101) -- (312,115) -- (284,115) -- cycle ;
\draw   (324,107) -- (352,107) -- (352,121) -- (324,121) -- cycle ;
\draw  (244,107) -- (272,107) -- (272,121) -- (244,121) -- cycle ;
\draw    (298,86.12) -- (298,98.12) ;
\draw [shift={(298,101.12)}, rotate = 270] [fill={rgb, 255:red, 0; green, 0; blue, 0 } ][line width=0.08]  [draw opacity=0] (3.57,-1.72) -- (0,0) -- (3.57,1.72) -- cycle;

\draw    (312,58) .. controls (321,60) and (338,62) .. (338,103) ;
\draw [shift={(338,106.56)}, rotate = 268.99] [fill={rgb, 255:red, 0; green, 0; blue, 0 }  ][line width=0.08]  [draw opacity=0] (3.57,-1.72) -- (0,0) -- (3.57,1.72) -- cycle;

\draw    (258,107) .. controls (258,62) and (276,60) .. (281,58) ;
\draw [shift={(284.3,58)}, rotate = 173.52] [fill={rgb, 255:red, 0; green, 0; blue, 0 }  ][line width=0.08]  [draw opacity=0] (3.57,-1.72) -- (0,0) -- (3.57,1.72) -- cycle;

\draw    (338,121.88) .. controls (334.62,137.16) and (273.79,138.05) .. (259,123.58) ;
\draw [shift={(258,121.43)}, rotate = 55.78] [fill={rgb, 255:red, 0; green, 0; blue, 0 }  ][line width=0.08]  [draw opacity=0] (3.57,-1.72) -- (0,0) -- (3.57,1.72) -- cycle;

\draw  [color={rgb, 255:red, 0; green, 0; blue, 0 } ,draw opacity=1 ][fill={rgb, 255:red, 255; green, 255; blue, 255 } ,fill opacity=1 ] (274.96,131.24) .. controls (274.97,130.43) and (276.04,129.28) .. (277.45,129.29) .. controls (278.86,129.30) and (279.91,130.43) .. (279.90,131.84) .. controls (279.89,133.25) and (278.82,134.40) .. (277.41,134.39) .. controls (276.00,134.38) and (274.95,133.25) .. (274.96,131.74) -- cycle ;

\draw    (280,131.24) .. controls (278.55,132.88) and (296.8,129.38) .. (297.8,115.38) ;
\draw    (268,72) .. controls (268.8,76) and (278.05,79) .. (284,78.88) ;
\draw    (312,79) .. controls (316,78) and (334,84) .. (337,91) ;

\draw  [color={rgb, 255:red, 0; green, 0; blue, 0 } ,draw opacity=1 ][fill={rgb, 255:red, 255; green, 255; blue, 255 } ,fill opacity=1 ] (334.05,91.34) .. controls (334.06,90.03) and (335.13,88.88) .. (336.54,88.89) .. controls (337.95,88.9) and (339,90.03) .. (338.99,91.44) .. controls (338.98,92.85) and (337.91,94.00) .. (336.50,93.99) .. controls (335.09,93.98) and (334.04,92.85) .. (334.05,91.34) -- cycle ;

\draw  [color={rgb, 255:red, 0; green, 0; blue, 0 } ,draw opacity=1 ][fill={rgb, 255:red, 255; green, 255; blue, 255 } ,fill opacity=1 ] (264.31,69.32) .. controls (264.32,68.51) and (265.39,67.36) .. (266.80,67.37) .. controls (268.21,67.38) and (269.26,68.51) .. (269.25,69.92) .. controls (269.24,71.33) and (268.17,72.48) .. (266.76,72.47) .. controls (265.35,72.46) and (264.30,71.33) .. (264.31,69.82) -- cycle ;


\draw (298.5,58.4) node [anchor=center][inner sep=0.75pt] [align=center] {\fontfamily{ptm}\selectfont {\scriptsize \QCVX}};
\draw (298.5,78.25) node [anchor=center][inner sep=0.75pt] [align=center] {\scriptsize {\fontfamily{ptm}\selectfont \MAX}};
\draw (298.5,108) node [anchor=center][inner sep=0.75pt] [align=center] {\fontfamily{ptm}\selectfont {\scriptsize \NI}};
\draw (339,114.5) node [anchor=center][inner sep=0.75pt] [align=center] {\fontfamily{ptm}\selectfont {\scriptsize \DPI}};
\draw (258,114.5) node [anchor=center][inner sep=0.75pt] [align=center] {\fontfamily{ptm}\selectfont {\scriptsize \MONO}};

\end{tikzpicture}}
        \caption{Implications of axioms. The merging arrows indicate joint implication: for example, on the left-hand side, we have that \MONO+\AVG~imply \CVX~\cite{2016alvimaxioms}.\label{fig:implications} } 
    \end{figure}
    \begin{figure}[htbp]
        \centering   
        \subfigure[\label{fig:CCV} Defining $\widehat{\mathbb{V}}$ as \AVG]{\tikzset{every picture/.style={line width=0.75pt}} 

\begin{tikzpicture}[x=0.75pt,y=0.75pt,yscale=-1,xscale=1]

\draw   (284,51) -- (312,51) -- (312,65) -- (284,65) -- cycle ;
\draw   (284,72) -- (312,72) -- (312,86) -- (284,86) -- cycle ;
\draw   (324,107) -- (352,107) -- (352,121) -- (324,121) -- cycle ;
\draw  (244,107) -- (272,107) -- (272,121) -- (244,121) -- cycle ;

\draw    (298,86.12) -- (298,116) ;


\draw    (312,58) .. controls (321,60) and (338,62) .. (338,103) ;
\draw [shift={(338,106.56)}, rotate = 268.99] [fill={rgb, 255:red, 0; green, 0; blue, 0 }  ][line width=0.08]  [draw opacity=0] (3.57,-1.72) -- (0,0) -- (3.57,1.72) -- cycle;

\draw    (258,107) .. controls (258,62) and (276,60) .. (281,58) ;
\draw [shift={(284.3,58)}, rotate = 173.52] [fill={rgb, 255:red, 0; green, 0; blue, 0 }  ][line width=0.08]  [draw opacity=0] (3.57,-1.72) -- (0,0) -- (3.57,1.72) -- cycle;

\draw    (338,121.88) .. controls (334.62,137.16) and (273.79,138.05) .. (259,123.58) ;
\draw [shift={(258,121.43)}, rotate = 55.78] [fill={rgb, 255:red, 0; green, 0; blue, 0 }  ][line width=0.08]  [draw opacity=0] (3.57,-1.72) -- (0,0) -- (3.57,1.72) -- cycle;

\draw    (280,131.24) .. controls (278.55,132.88) and (296.8,129.38) .. (298,116) ;
\draw    (268,72) .. controls (268.8,76) and (278.05,79) .. (284,78.88) ;
\draw    (312,79) .. controls (316,78) and (334,84) .. (337,91) ;

\draw  [color={rgb, 255:red, 0; green, 0; blue, 0 } ,draw opacity=1 ][fill={rgb, 255:red, 255; green, 255; blue, 255 } ,fill opacity=1 ] (274.96,131.24) .. controls (274.97,130.43) and (276.04,129.28) .. (277.45,129.29) .. controls (278.86,129.30) and (279.91,130.43) .. (279.90,131.84) .. controls (279.89,133.25) and (278.82,134.40) .. (277.41,134.39) .. controls (276.00,134.38) and (274.95,133.25) .. (274.96,131.74) -- cycle ;

\draw  [color={rgb, 255:red, 0; green, 0; blue, 0 } ,draw opacity=1 ][fill={rgb, 255:red, 255; green, 255; blue, 255 } ,fill opacity=1 ] (334.05,91.34) .. controls (334.06,90.03) and (335.13,88.88) .. (336.54,88.89) .. controls (337.95,88.9) and (339,90.03) .. (338.99,91.44) .. controls (338.98,92.85) and (337.91,94.00) .. (336.50,93.99) .. controls (335.09,93.98) and (334.04,92.85) .. (334.05,91.34) -- cycle ;

\draw  [color={rgb, 255:red, 0; green, 0; blue, 0 } ,draw opacity=1 ][fill={rgb, 255:red, 255; green, 255; blue, 255 } ,fill opacity=1 ] (264.31,69.32) .. controls (264.32,68.51) and (265.39,67.36) .. (266.80,67.37) .. controls (268.21,67.38) and (269.26,68.51) .. (269.25,69.92) .. controls (269.24,71.33) and (268.17,72.48) .. (266.76,72.47) .. controls (265.35,72.46) and (264.30,71.33) .. (264.31,69.82) -- cycle ;

\draw (298.5,58.4) node [anchor=center][inner sep=0.75pt] [align=center] {\fontfamily{ptm}\selectfont {\scriptsize \CCVX}};
\draw (298.5,78.25) node [anchor=center][inner sep=0.75pt] [align=center] {\scriptsize {\fontfamily{ptm}\selectfont \GAVG}};
\draw (339,114.5) node [anchor=center][inner sep=0.75pt] [align=center] {\fontfamily{ptm}\selectfont {\scriptsize \DPI}};
\draw (258,114.5) node [anchor=center][inner sep=0.75pt] [align=center] {\fontfamily{ptm}\selectfont {\scriptsize \CIV}};

\end{tikzpicture}}
        \subfigure[\label{fig:CCVMAX} Defining $\widehat{\mathbb{V}}$ as \MAX]{\tikzset{every picture/.style={line width=0.75pt}} 

\begin{tikzpicture}[x=0.75pt,y=0.75pt,yscale=-1,xscale=1]

\draw   (284,51) -- (312,51) -- (312,65) -- (284,65) -- cycle ;
\draw   (284,72) -- (312,72) -- (312,86) -- (284,86) -- cycle ;
\draw   (324,107) -- (352,107) -- (352,121) -- (324,121) -- cycle ;
\draw  (244,107) -- (272,107) -- (272,121) -- (244,121) -- cycle ;
\draw    (298,86.12) -- (298,116) ;

\draw    (312,58) .. controls (321,60) and (338,62) .. (338,103) ;
\draw [shift={(338,106.56)}, rotate = 268.99] [fill={rgb, 255:red, 0; green, 0; blue, 0 }  ][line width=0.08]  [draw opacity=0] (3.57,-1.72) -- (0,0) -- (3.57,1.72) -- cycle;

\draw    (258,107) .. controls (258,62) and (276,60) .. (281,58) ;
\draw [shift={(284.3,58)}, rotate = 173.52] [fill={rgb, 255:red, 0; green, 0; blue, 0 }  ][line width=0.08]  [draw opacity=0] (3.57,-1.72) -- (0,0) -- (3.57,1.72) -- cycle;

\draw    (338,121.88) .. controls (334.62,137.16) and (273.79,138.05) .. (259,123.58) ;
\draw [shift={(258,121.43)}, rotate = 55.78] [fill={rgb, 255:red, 0; green, 0; blue, 0 }  ][line width=0.08]  [draw opacity=0] (3.57,-1.72) -- (0,0) -- (3.57,1.72) -- cycle;

\draw  [color={rgb, 255:red, 0; green, 0; blue, 0 } ,draw opacity=1 ][fill={rgb, 255:red, 255; green, 255; blue, 255 } ,fill opacity=1 ] (274.96,131.24) .. controls (274.97,130.43) and (276.04,129.28) .. (277.45,129.29) .. controls (278.86,129.30) and (279.91,130.43) .. (279.90,131.84) .. controls (279.89,133.25) and (278.82,134.40) .. (277.41,134.39) .. controls (276.00,134.38) and (274.95,133.25) .. (274.96,131.74) -- cycle ;

\draw    (280,131.24) .. controls (278.55,132.88) and (296.8,129.38) .. (298,116) ;
\draw    (268,72) .. controls (268.8,76) and (278.05,79) .. (284,78.88) ;
\draw    (312,79) .. controls (316,78) and (334,84) .. (337,91) ;

\draw  [color={rgb, 255:red, 0; green, 0; blue, 0 } ,draw opacity=1 ][fill={rgb, 255:red, 255; green, 255; blue, 255 } ,fill opacity=1 ] (334.05,91.34) .. controls (334.06,90.03) and (335.13,88.88) .. (336.54,88.89) .. controls (337.95,88.9) and (339,90.03) .. (338.99,91.44) .. controls (338.98,92.85) and (337.91,94.00) .. (336.50,93.99) .. controls (335.09,93.98) and (334.04,92.85) .. (334.05,91.34) -- cycle ;

\draw  [color={rgb, 255:red, 0; green, 0; blue, 0 } ,draw opacity=1 ][fill={rgb, 255:red, 255; green, 255; blue, 255 } ,fill opacity=1 ] (264.31,69.32) .. controls (264.32,68.51) and (265.39,67.36) .. (266.80,67.37) .. controls (268.21,67.38) and (269.26,68.51) .. (269.25,69.92) .. controls (269.24,71.33) and (268.17,72.48) .. (266.76,72.47) .. controls (265.35,72.46) and (264.30,71.33) .. (264.31,69.82) -- cycle ;


\draw (298.5,58.4) node [anchor=center][inner sep=0.75pt] [align=center] {\fontfamily{ptm}\selectfont {\scriptsize \QCVX}};
\draw (298.5,78.25) node [anchor=center][inner sep=0.75pt] [align=center] {\scriptsize {\fontfamily{ptm}\selectfont \MAX}};
\draw (339,114.5) node [anchor=center][inner sep=0.75pt] [align=center] {\fontfamily{ptm}\selectfont {\scriptsize \DPI}};
\draw (258,114.5) node [anchor=center][inner sep=0.75pt] [align=center] {\fontfamily{ptm}\selectfont {\scriptsize \CIV}};

\end{tikzpicture}}
        \caption{Implications of axioms in~\cite{2020CondEntropyAxiom}. \label{fig:CCVaxiom}} 
    \end{figure}
    
    It has been shown that both \AVG~and \MAX~definitions of posterior $g$-vulnerability satisfy the \NI~axiom. Then, for \AVG, the axioms of \CVX, \MONO, and \DPI~axioms are equivalent, and for \MAX, the \QCVX, \MONO, and \DPI~ are equivalent. These results are shown in Fig.~\ref{fig:implications}~\cite[Fig. 2]{2016alvimaxioms}. 
    The same axioms hold for uncertainty (entropy) measures by replacing convexity with concavity and quasi-convexity with 
    quasi-concavity. Also, note that the \MAX~will be replaced with \MIN.

    In~\cite{2020CondEntropyAxiom}, the convexity axiom was relaxed to core-convexity (\CCVX)\footnote{In~\cite{2020CondEntropyAxiom}, uncertainty rather than vulnerability was considered. We replaced vulnerability axioms for the sake of consistency.}, 
    and the standard notion of averaging was generalized to \textit{generalized averaging} (\GAVG). 
    Additionally, the axioms of \NI~and \MONO~were combined into a single axiom called \textit{conditioning increases vulnerability} (\CIV), resulting in the revised set of axioms illustrated in Figure~\ref{fig:CCVaxiom}. 
    This axiomatization was further developed in~\cite{2021ConvCorConV}, where a broader framework for posterior vulnerabilities was introduced. 
    This framework employs a limit construction over sequences of core-convex vulnerabilities, demonstrating that quasiconcave functions emerge as such limits. 
    The generalized framework eliminates the dichotomy between \AVG~and \MAX~by encompassing both within a unified set of axioms. 
    
    Despite the generalizations proposed in~\cite{2020CondEntropyAxiom,2021ConvCorConV}, we adhere to the conventional set of axioms introduced in~\cite{2016alvimaxioms} and retain the dichotomy for simplicity. 
    While core-convexity is a reasonable assumption for vulnerability measures, it does not inherently guarantee convexity, which may be essential for certain applications. 
    
    \subsection{Axioms of Generalized Vulnerabilities}

        We first study the axiomatic relations of the generalized prior vulnerability.
        
        \begin{theorem}\label{thm:Vfg axioms}
            $V_{f,g}(\pi)$ satisfies axioms of prior vulnerability. 
        \end{theorem}  
        \begin{proof}
            We follow similar steps in~\cite[Sec. IV.A]{2016alvimaxioms}. 
            Let $$f\!\circ\!{g}_{w}(\pi)=f^{-1}\bigg(\sum_{x\in\X}\pi_{x}f\big(g(w,x)\big)\bigg),$$ 
            which is the generalized expected gain for a specific guess $w$. 
            Consider $\pi=\sum_{i}a_{i}\pi^{i}$ for some priors $\pi^{1},\cdots, \pi^{n}$ and non-negative reals $a_{1},a_{2},\cdots,a_{n}$ such that $\sum_{i}a_i=1$. Due to the convexity of $f^{-1}$ (assumed in def.~\ref{def:generalprior}), we have:
            \begin{align*}
                &f\!\circ\!{g}_{w}(\pi)=f^{-1}\bigg( \sum_{x\in\X} \Big(\sum_{i}a_{i}\pi^{i}_{x}\Big)f\big(g(w,x)\big)\bigg)\\
                &= f^{-1}\bigg(\sum_{i}a_{i}\sum_{x\in\X}\pi^{i}_{x}f\big(g(w,x)\big)\bigg)\\
                &\leq \sum_{i} a_{i}f^{-1}\bigg(\sum_{x\in\X}\pi^{i}_{x}f\big(g(w,x)\big)\bigg)
                =\sum_{i} a_{i} f\!\circ\!g_{w}(\pi^{i}).
            \end{align*}
            Therefore, $f\!\circ\!{g}_{w}(\pi)$ is convex and continuous w.r.t to  $\pi$. Accordingly, $V_{f,g}(\pi)=\sup_{w\in\W}f\!\circ\!{g}_{w}(\pi)$ is the supremum over a family of convex and continuous functions, so it is convex and continuous due to the argument in~\cite[Prop. 2, Thm. 3]{2016alvimaxioms}. 
            The convexity of $V_{f,g}(\pi)$ implies  \textbf{quasi-convexity}. 
        \end{proof}
        \begin{remark}
             For the generalized uncertainty measure, $f^{-1}$ should be concave to keep the concavity of $U_{f,l}(\pi)$. 
        \end{remark}
        \begin{remark}
            Note that the assumption of $f^{-1}$ convexity is a sufficient condition that makes the proof straightforward and may not be necessary. Thus, if the convexity of $f^{-1}$ can be relaxed and only the basic properties of continuity and strict monotonicity are used, a stronger result may be obtained.
        \end{remark}

        Now, we study the axiomatic relations of the generalized posterior vulnerability.
        \begin{proposition}[{\AVG$\Rightarrow$\NI}]
            If a pair of generalized prior/posterior vulnerabilities $\left({V}_{f,g}, \widehat{V}_{h,f,g}\right)$ satisfy \AVG~then they also satisfy \NI.
        \end{proposition}
        \begin{proof}
            For a \NI~channel $C$ we have $\delta^{y}=\pi, \forall y \in \Y$. Thus, 
            \begin{align*}
            \widehat{V}_{h,f,g}[\pi, \Bar{0}] &=h^{-1}\bigg(\sum_{y\in\Y}p(y)h\big(V_{f,g}(\pi)\big) \bigg) \\
            &= h^{-1}\bigg(h \big( V_{f,g}(\pi) \big) \sum_{y\in\Y}p(y) \bigg) 
            =V_{f,g}(\pi). \quad~~~\qedhere
            \end{align*}
        \end{proof}

        \begin{proposition}[\NI+\DPI$\Rightarrow$\MONO] 
            If a pair of generalized prior/posterior vulnerabilities \!$({V}_{f,g}, \widehat{V}_{h,f,g})$ satisfy \NI~and \DPI, then they also satisfy \MONO.
        \end{proposition}
        
        \begin{proof}
            For any $[\pi, C]$, let $\overline{0}$ denote a \NI~channel with one column and as many rows as the columns of $C$, then
            \begin{align*}
                \widehat{V}_{h,f,g}[\pi, C] \geq \widehat{V}_{h,f,g}[\pi, C \overline{0}] 
                =  \widehat{V}_{h,f,g}[\pi, \overline{0}] 
                =  {V}_{f,g}(\pi).
            \end{align*}
            The inequality is due to \DPI~and $C\overline{0}=0$.
        \end{proof}

        \begin{remark}
            In~\cite[Prop. 8]{2016alvimaxioms}, it was shown that if a pair of prior/posterior vulnerabilities satisfy \AVG~and \MONO, it implies \CVX~for the prior vulnerability. We decided not to include this property because it seems unnecessary, given that the convexity of the prior vulnerability $\mathbb{V}$ is already assumed when we define $\widehat{\mathbb{V}}$. Additionally, proving this property in general case would require $h^{-1}$ to be convex if $h\neq f$, while for the \DPI, we need it to be concave, resulting in an affine $h$ that is not useful for a generalized definition of posterior vulnerability. This can be considered our relaxation adhering to convexity instead of core-convexity.
        \end{remark}
        
        \begin{proposition}[\AVG+\CVX$\Rightarrow$\DPI]\label{prop:AVG+CVX=DPI}
            If a pair of prior/posterior vulnerabilities $({V}_{f,g}, \widehat{V}_{h,f,g})$ satisfy \AVG~and \CVX, then they also satisfy \DPI.\\
            The proof is provided in Appendix~\ref{app:proofofAVG+CVX=DPI}.
        \end{proposition}

        We now prove the axiomatic relations for the maximum posterior vulnerability.
        
            \begin{proposition}[{\MAX$\Rightarrow$\NI}] 
                If a pair of generalized prior/posterior vulnerabilities $(V_{f,g},\widehat{V}_{f,g}^{\max})$ satisfy \MAX, they also satisfy \NI.
            \end{proposition}
            
            \renewcommand{\IEEEproofindentspace}{0pt}
            \begin{proof}
                For a \NI~channel, $\delta^{y}=\pi, \forall y \in \Y$, thus we have:
               \[\widehat{V}^{\max}_{f,g}[\pi]=\max_{y\in\Y}V_{f,g}(\pi)=V_{f,g}(\pi).~\qedhere\]
            \end{proof}
            \begin{proposition}[\MAX+\QCVX $\Rightarrow$\DPI] 
                If a pair of generalized prior/posterior vulnerabilities $(V_{f,g},\widehat{V}_{f,g}^{\max})$ satisfy \MAX~and \QCVX, they also satisfy \DPI. 
            \end{proposition}
            
            \begin{proof}
                Consider a Markov chain similar to  Proposition~\ref{prop:AVG+CVX=DPI}. 
                \begin{align*}
                     \widehat{V}^{\max}_{f,g}[\pi,CR]&=\max_{z}V_{f,g}\left(p(X|z)\right) \\
                     &=\max_{z}V_{f,g}\bigg(\sum_{y\in\Y}p(x|y)p(y|z)\bigg) \\
                     &\leq\max_{z}\left(\max_{y\in\Y} V_{f,g}\left(p(x|y)\right)\right)=\widehat{V}^{\max}_{f,g}[\pi,C]. ~\qedhere
                \end{align*}
            \end{proof}
            Note that we dropped the implication of the convexity of the prior vulnerability from the posterior. This gives somewhat different relationships between our axioms. See Fig.~\ref{fig:generalimplications}.
    
             \begin{figure}[t]
                \centering   
                \subfigure[\label{fig:GAVGa} Axioms of $\widehat{V}_{h,f,g}$]{\tikzset{every picture/.style={line width=0.75pt}} 

\begin{tikzpicture}[x=0.75pt,y=0.75pt,yscale=-1,xscale=1]

\draw   (244,51) -- (272,51) -- (272,65) -- (244,65) -- cycle ;
\draw   (324,51) -- (352,51) -- (352,65) -- (324,65) -- cycle ;
\draw   (244,90) -- (272,90) -- (272,104) -- (244,104) -- cycle ;
\draw   (324,90) -- (352,90) -- (352,104) -- (324,104) -- cycle ;
 \draw  (284,120) -- (312,120) -- (312,134) -- (284,134) -- cycle ;

\draw (339,58.4) node [anchor=center][inner sep=0.75pt] [align=center] {\fontfamily{ptm}\selectfont {\scriptsize \CVX}};
\draw (258,58.4) node [anchor=center][inner sep=0.75pt] [align=center] {\scriptsize {\fontfamily{ptm}\selectfont \GAVG}};
\draw (258,97) node [anchor=center][inner sep=0.75pt] [align=center] {\fontfamily{ptm}\selectfont {\scriptsize \NI}};
\draw (339,97) node [anchor=center][inner sep=0.75pt] [align=center] {\fontfamily{ptm}\selectfont {\scriptsize \DPI}};
\draw (298,127) node [anchor=center][inner sep=0.75pt] [align=center] {\fontfamily{ptm}\selectfont {\scriptsize \MONO}};

\draw    (258,65) -- (258,87) ;
\draw [shift={(258,90)}, rotate = 270] [fill={rgb, 255:red, 0; green, 0; blue, 0 } ][line width=0.08]  [draw opacity=0] (3.57,-1.72) -- (0,0) -- (3.57,1.72) -- cycle;

\draw    (339,65) -- (339,87) ;
\draw [shift={(339,90)}, rotate = 270] [fill={rgb, 255:red, 0; green, 0; blue, 0 } ][line width=0.08]  [draw opacity=0] (3.57,-1.72) -- (0,0) -- (3.57,1.72) -- cycle;

\draw    (272,97) -- (324,97) ;

\draw    (298,97) -- (298,118) ;
\draw [shift={(298,120)}, rotate = 270] [fill={rgb, 255:red, 0; green, 0; blue, 0 } ][line width=0.08]  [draw opacity=0] (3.57,-1.72) -- (0,0) -- (3.57,1.72) -- cycle;

\draw  (272,58.4) .. controls (300,58.4) and (300.05,75) .. (339,75) ;
\draw [color={rgb, 255:red, 0; green, 0; blue, 0 }, draw opacity=1] [fill={rgb, 255:red, 255; green, 255; blue, 255 }, fill opacity=1] (298, 97) circle (2);

\draw [color={rgb, 255:red, 0; green, 0; blue, 0 }, draw opacity=1] [fill={rgb, 255:red, 255; green, 255; blue, 255 }, fill opacity=1] (339, 75) circle (2);

\end{tikzpicture}}
                \subfigure[\label{fig:GAVG} Axioms of $\widehat{V}^{\max}_{f,g}$]{\tikzset{every picture/.style={line width=0.75pt}} 

\begin{tikzpicture}[x=0.75pt,y=0.75pt,yscale=-1,xscale=1]

\draw   (244,51) -- (272,51) -- (272,65) -- (244,65) -- cycle ;
\draw   (324,51) -- (352,51) -- (352,65) -- (324,65) -- cycle ;
\draw   (244,90) -- (272,90) -- (272,104) -- (244,104) -- cycle ;
\draw   (324,90) -- (352,90) -- (352,104) -- (324,104) -- cycle ;
 \draw  (284,120) -- (312,120) -- (312,134) -- (284,134) -- cycle ;

\draw (339,58.4) node [anchor=center][inner sep=0.75pt] [align=center] {\fontfamily{ptm}\selectfont {\scriptsize \QCVX}};
\draw (258,58.4) node [anchor=center][inner sep=0.75pt] [align=center] {\scriptsize {\fontfamily{ptm}\selectfont \GMAX}};
\draw (258,97) node [anchor=center][inner sep=0.75pt] [align=center] {\fontfamily{ptm}\selectfont {\scriptsize \NI}};
\draw (339,97) node [anchor=center][inner sep=0.75pt] [align=center] {\fontfamily{ptm}\selectfont {\scriptsize \DPI}};
\draw (298,127) node [anchor=center][inner sep=0.75pt] [align=center] {\fontfamily{ptm}\selectfont {\scriptsize \MONO}};

\draw    (258,65) -- (258,87) ;
\draw [shift={(258,90)}, rotate = 270] [fill={rgb, 255:red, 0; green, 0; blue, 0 } ][line width=0.08]  [draw opacity=0] (3.57,-1.72) -- (0,0) -- (3.57,1.72) -- cycle;

\draw    (339,65) -- (339,87) ;
\draw [shift={(339,90)}, rotate = 270] [fill={rgb, 255:red, 0; green, 0; blue, 0 } ][line width=0.08]  [draw opacity=0] (3.57,-1.72) -- (0,0) -- (3.57,1.72) -- cycle;

\draw    (272,97) -- (324,97) ;

\draw    (298,97) -- (298,118) ;
\draw [shift={(298,120)}, rotate = 270] [fill={rgb, 255:red, 0; green, 0; blue, 0 } ][line width=0.08]  [draw opacity=0] (3.57,-1.72) -- (0,0) -- (3.57,1.72) -- cycle;

\draw  (272,58.4) .. controls (300,58.4) and (300.05,75) .. (339,75) ;
\draw [color={rgb, 255:red, 0; green, 0; blue, 0 }, draw opacity=1] [fill={rgb, 255:red, 255; green, 255; blue, 255 }, fill opacity=1] (298, 97) circle (2);

\draw [color={rgb, 255:red, 0; green, 0; blue, 0 }, draw opacity=1] [fill={rgb, 255:red, 255; green, 255; blue, 255 }, fill opacity=1] (339, 75) circle (2);

\end{tikzpicture}}
                \caption{Implications of axioms for generalized prior and posterior vulnerabilities. 
                The merging arrows indicate joint implication: for example, in~\ref{fig:GAVG}, we have \AVG+\CVX~imply \DPI and \NI+\DPI~imply~\MONO.\label{fig:generalimplications} } 
            \end{figure}

\section{Conclusion}\label{sec:conclusion}
    In this paper, we introduced a generalized QIF framework based on the Kolmogorov-Nagumo \(\!f\!\)-mean to bridge the gap between the traditional QIF framework and \(\!\alpha\)-based leakage measures, along with their maximal versions from information-theoretic privacy. While \(\!\alpha\)-based measures have been defined in a somewhat similar fashion to the $g$-leakage model in QIF, they presented inconsistencies with the axiomatic approach of QIF. Our generalized framework resolved these issues, offering a consistent interpretation of the operational meaning of all these measures within the extended QIF framework. A key result was demonstrating the equivalence between maximal leakage, its generalized form, and the generalized capacity measure, which simplified the interpretation and addressed complexities related to guessing randomized functions of the secret in the maximal leakage model.

    {\color{black} This framework and the core-concave approach could be extended for future work to include other leakage measures and their corresponding gain or loss functions, 
    such as total variation distance~\cite{2019PUTTotalDistnce,2022DataDsclsurell1Priv}, $\chi^2$-divergence~\cite{2019PrivEstimGuarant,2021StrongChi2}, and other \(f\)-divergences.} Additionally, exploring worst-case capacity leakage measures for more general \(f\)-mean functions and identifying minimal conditions for \(f\)-mean functions to satisfy vulnerability axioms can be interesting directions for research.

\bibliographystyle{IEEEtran}
\bibliography{IEEEabrv,BIBAxiom}

\appendices

\section{Proof of Theorem~\ref{thm:maximal=capcity}}\label{appen:maximal=capcity}
 
    We first need to prove the following lemma.
    \begin{lemma}\label{lem:infp(u|x)} 
        For a given \(f\), any randomized function $U$ of secret $X$, and any gain function $g:\W \times \U \rightarrow \Real$, we have:
        \begin{align}          
        &\inf_{p_{U|X}}\sum_{x,u}p(x,u)f\big(g(w,u)\big) = \sum_{x\in\X}\pi_{x} f\big(g(w,x)\big), \label{eq:infpux}\\
        & \sup_{p_{U|X}}\sum_{x,u}p(x,u)f\big(g(w,u)\big) = \sum_{x\in\X}\pi_{x} f\big(g(w,x)\big).\label{eq:supux}
        \end{align}
    \end{lemma}
    \begin{proof} 
        We prove~\eqref{eq:infpux} by showing the RHS is both the upper and lower bound of the LHS. 
        Equation~\eqref{eq:supux} is proven similarly.
        
        Consider the following distribution:
        \begin{align}\label{eq:qu|x}
            q_{U|X}(u|x)=
            \begin{cases}
                1, &u=x,\\
                0,  &u\neq x.
            \end{cases}
        \end{align}
        Then we have:
        
        \begin{align*}
            &\inf_{p_{U|X}} \sum_{x,u}p(x,u)\!f\big(g(w,u)\big) \\
            &\leq \sum_{x\in\X}\!\pi_{x}\!\sum_{u\in\U}q_{U|X}(u|x)f\big(g(w,u)\big) 
            =\sum_{x\in\X}\pi_{x}f(g(w,x)).
        \end{align*}
        For any randomized function $U$ and without loss of generality, let $\U$ be defined as $\U=\bigcup_{x\in \suport{\pi}}\{(x,u_{x}):u_{x} \in \{1,2,\cdots,k(x)\}\}$, where $k(x) \geq 1$.  Similar to~\cite{2020MaxL}, any gain function $g(w,u)$ can be written as $g\big(w,(x,u_{x})\big)$. We assume the randomized function $U$ is a surjective  function of $x\in\suport{\pi}$.
        \begin{align}       
            &\inf_{p_{U|X}}\sum_{x,u}p(x,u)f\big(g(w,u)\big) \nonumber\\
            &=\inf_{p_{U|X}}\sum_{x\in\X}\pi_{x}\sum_{u\in\U}p(u|x) f\Big(g\big(w,(x,u_x)\big)\Big) \nonumber\\
            &\geq\sum_{x\in\X}\pi_{x} \inf_{p_{U|x}}\sum_{u\in\U}p(u|x) f\Big(g\big(w,(x,u_x)\big)\Big) \nonumber\\
            &=\sum_{x\in\X}\pi_{x} f\big(g(w,x)\big) \label{eq:g(w,u)=g(w,x)},
        \end{align}
        where~\eqref{eq:g(w,u)=g(w,x)} is  given by the following $p_{U|x}$
        \begin{align}
            p(u|x)=
            \begin{cases}
                1, & u\in \argmin_{u_x}f\Big(g\big(w,(x,u_x)\big)\Big), \nonumber\\
                0, &\text{otherwise}. 
            \end{cases} 
            \qquad~~\qedhere
        \end{align} 
    \end{proof}
    
    Now we prove Theorem~\ref{thm:maximal=capcity}. The distribution in~\eqref{eq:qu|x} is used to show $g$-leakage capacity is a lower bound on the generalized maximal leakage.
    
    \begin{align*}
        & \sup_{U-X-Y}\Lk^{\times}_{h,f,g}(p_{U},C)=\log \sup_{U-X-Y} \frac{\widehat{V}_{h,f,g}[p_{U},C]}{V_{f,g}(p_{U})} \\
        &=\log\sup_{\pi}\sup_{p_{U|X}}\frac{ \displaystyle   h^{-1}\bigg( \sum_{y\in\Y} p(y) h\Big(V_{f,g}\big(p_{U|y}\big)\Big)\bigg)}{ \displaystyle  \sup_{w\in\W}f^{-1}\Big( \sum_{u\in\U}p(u) f\big(g(w,u)\big)\Big)} \nonumber\\         
        &=\log \sup_{\pi}\sup_{p_{U|X}} 
        \frac{ \displaystyle  h^{-1}\!\Bigg(\sum_{y\in\Y} p(y) h\!\bigg(V_{f,g}\Big(\sum_{x\in\X}{p(u|x)p(x|y)}\!\Big)\!\bigg)\!\Bigg)}{ \displaystyle  \sup_{w\in\W}f^{-1}\bigg( \sum_{u\in\U}\Big(\sum_{x\in\X}\pi_{x}p(u|x)\Big) f\big(g(w,u)\big)\!\bigg)} \nonumber\\
        &\geq \log \sup_{\pi}\frac{ \displaystyle   h^{-1}\Bigg( \sum_{y\in\Y} p(y)h\bigg(V_{f,g}\Big(\sum_{x\in\X}{q_{U|X}(u|x)p(x|y)}\!\Big)\!\bigg)\!\Bigg)}{ \displaystyle  \sup_{w\in\W}f^{-1}\bigg( \sum_{u\in\U}\Big(\sum_{x\in\X}\pi_{x}q_{U|X}(u|x)\Big) f\big(g(w,u)\big)\!\bigg)} \nonumber\\
        &=\log \sup_{\pi}\frac{ \displaystyle   h^{-1}\Bigg(\sum_{y\in\Y} p(y) h\bigg(V_{f,g}\Big(\sum_{x\in\X}{q(u|x)p(x|y)}\Big)\!\bigg)\!\Bigg)}{ \displaystyle   \sup_{w\in\W}f^{-1}\Big( \sum_{x\in\X}\pi_{x}\sum_{u\in\U}q_{U|X}(u|x)f\big(g(w,u)\big)\!\Big)} \nonumber\\
        &= \log\sup_{\pi}\frac{ \displaystyle   h^{-1}\bigg( \sum_{y\in\Y} p(y) h\Big(V_{f,g}\big(p_{X|y}\big)\Big)\!\bigg)}{ \displaystyle   \sup_{w\in\W}f^{-1}\Big( \sum_{x\in\X}\pi_{x}f\big(g(w,x)\big)\!\Big)}= \sup_{\pi}\mathcal{L}_{h,f,g}^{\times}(\pi,C)
    \end{align*} 
    For the upper bound, we write
    
    \begin{align*}
        & \sup_{U-X-Y}\Lk^{\times}_{h,f,g}\big(p(U),C\big)=\log \sup_{U-X-Y} \frac{\widehat{V}_{h,f,g}[p_{U},C]}{V_{f,g}(p_{U})} \nonumber\\
        &=\log \sup_{\pi}\sup_{p_{U|X}}\frac{ \displaystyle  h^{-1}\bigg( \sum_{y\in\Y} p(y)h\Big(V_{f,g}\big(p_{U|y}\big)\Big) \bigg)}{ \displaystyle \sup_{w\in\W}f^{-1}\Big( \sum_{u\in\U}p(u) f\big(g(w,u)\big)\Big)} \nonumber\\
        &\leq \log \sup_{\pi}\frac{  \displaystyle \sup_{p_{U|X}} h^{-1}\bigg( \sum_{y\in\Y} p(y) h\Big( V_{f,g}\big(p_{U|y}\big)\Big)\bigg)}{ \displaystyle  \inf_{p_{U|X}} \sup_{w\in\W}f^{-1}\Big( \sum_{u\in\U}p(u) f\big(g(w,u)\big)\Big)}\\
        &\leq \log \sup_{\pi}\frac{ \displaystyle  \displaystyle  \sup_{p_{U|X}} h^{-1}\bigg( \sum_{y\in\Y} p(y) h\Big(V_{f,g}\big(p_{U|y}\big)\Big)\bigg)}{ \displaystyle   \sup_{w\in\W} \inf_{p_{U|X}}f^{-1}\Big( \sum_{u\in\U}p(u) f\big(g(w,u)\big)\Big)}\\
        &= \log \sup_{\pi}\frac{ \displaystyle   \sup_{p_{U|X}} h^{-1}\bigg( \sum_{y\in\Y} p(y) h\Big( V_{f,g}\big(p_{U|y}\big)\Big)\bigg)}{ \displaystyle   \sup_{w\in\W} f^{-1}\bigg(\inf_{p_{U|X}}\sum_{u\in\U}\Big(\sum_{x\in\X}\pi_{x}p(u|x)\Big)f\big(g(w,u)\big) \bigg)} \nonumber\\
        &=\log \sup_{\pi}\frac{ \displaystyle  \sup_{p_{U|X}} h^{-1}\bigg( \sum_{y\in\Y} p(y) h\Big(V_{f,g}\big(p_{U|y}\big)\Big)\bigg)}{\displaystyle\sup_{w\in\W} f^{-1}\Big(\inf_{p_{U|X}}\sum_{u\in\U}\sum_{x\in\X}\pi_{x}p(u|x) f\big(g(w,u)\big)\Big)} \nonumber\\
        &\leq \log \sup_{\pi}\frac{ \displaystyle  \sup_{p_{U|X}} h^{-1}\bigg( \sum_{y\in\Y} p(y) h\Big( V_{f,g}\big(p_{U|y}\big)\Big)\bigg)}{\displaystyle\!\sup_{w\in\W} f^{-1}\Big(\!\sum_{x\in\X} \!\pi_{x}\!\inf_{p_{U|X}} \!\sum_{u\in\U} \! p(u|x) f\big(g(w,u)\big) \!\Big)  } \\
        &=\log \sup_{\pi}\frac{ \displaystyle   h^{-1}\bigg( \sum_{y\in\Y} p(y) h\Big(V_{f,g}\big(p_{X|y}\big)\Big)\bigg) }{  \displaystyle \sup_{w\in\W} f^{-1}\Big( \sum_{x\in\X}\pi_{x}f\big(g(w,x)\big) \Big) }= \sup_{\pi}\mathcal{L}_{h,f,g}^{\times}(\pi,C).\nonumber ~~\qedhere
    \end{align*}

\section{Proof of Proposition~\ref{prop:maximalalphabeta} }\label{app:proofofalphabeta}

 \begin{proof} 
    For the given functions, we have:
    
    \begin{align*}
        \widehat{V}_{h_{(\alpha,\beta)},f_{\alpha},g}[\pi,C] &=h_{(\alpha,\beta)}^{-1}\left( \sum_{y\in\Y} p(y) h_{(\alpha,\beta)}\left( V_{f_{\alpha},g}\left(p_{X|y}\right)\right) \right) \nonumber\\
                &=\left( \sum_{y\in\Y}p(y) \left(\sum_{x\in\X}p({x|y})\right)^{\frac{\beta}{\alpha}} \right)^{\frac{\alpha}{(\alpha-1)\beta}}.
    \end{align*}
    And the corresponding leakage  $\Lk^{\times}_{h_{(\alpha,\beta)},f_{\alpha},g}\left(\pi,C\right)$ is 
    
    \begin{align*}
        &\Lk^{\times}_{h_{(\alpha,\beta)},f_{\alpha},g}\left(\pi,C\right)= \log\frac{\widehat{V}_{h_{(\alpha,\beta)},f_{\alpha},g}[\pi,C]}{V_{f_{\alpha},g}(\pi)} \nonumber\\
        &=\log \frac{\left( \sum_{y\in\Y}p(y) \left(\sum_{x\in\X}p^{\alpha}({x|y})\right)^{\frac{\beta}{\alpha}} \right)^{\frac{\alpha}{(\alpha-1)\beta}}}{\left(\sum_{x\in\X}\pi_{x}^{\alpha}\right)^{\frac{1}{\alpha-1}}} \nonumber\\
        &=\frac{\alpha}{\alpha-1}\log \frac{\left( \sum_{y\in\Y}p(y) \left(\sum_{x\in\X}p^{\alpha}({x|y})\right)^{\frac{\beta}{\alpha}} \right)^{\frac{1}{\beta}}}{\left(\sum_{x\in\X}\pi_{x}^{\alpha}\right)^{\frac{1}{\alpha}}} \nonumber\\
        &=\frac{\alpha}{\alpha-1}\log \frac{\left( \sum_{y\in\Y}p(y)^{1-\beta} \left(\sum_{x\in\X}p^{\alpha}({x,y})\right)^{\frac{\beta}{\alpha}} \right)^{\frac{1}{\beta}}}{\left(\sum_{x\in\X}\pi_{x}^{\alpha}\right)^{\frac{1}{\alpha}}} \nonumber\\
        &=\frac{\alpha}{(\alpha-1)\beta}\log  \sum_{y\in\Y}p(y)^{1-\beta} \left[ \frac{\sum_{x\in\X}\pi_{x}^{\alpha}C_{x,y}^{\alpha}}{{\sum_{x\in\X}\pi_{x}^{\alpha}}} \right]^{\frac{\beta}{\alpha}}. \nonumber ~~~~\qedhere
    \end{align*}
\end{proof}
        
\section{Proof of Equation~\ref{eq:localreyni}}\label{app:proofoflocalreyni}
    By~\eqref{eq:alphabeatleak} and~\eqref{eq:alphabetacap}, for \(\alpha=\beta\) we obtain:
    \begin{align}
    &\sup_{\pi}\Lk^{\times}_{h_{(\alpha,\alpha)},f_{\alpha},g}(\pi,C)  \nonumber\\
    &=\sup_{\pi} \frac{1}{\alpha-1}\log {\mathlarger{\sum_{y\in\Y}}}{p(y)^{1-\alpha}\left[ \sum_{x\in\X}\frac{\pi_{x}^{\alpha}C_{x,y}^{\alpha}}{{\sum_{x\in\X}\pi_{x}^{\alpha}}} \right]} \nonumber\\
    &=\sup_{\pi} \frac{1}{\alpha-1}\log {\mathlarger{\sum_{y\in\Y}}}\left(\sum_{x'}\pi_{x'}C_{x',y}\right)^{1-\alpha}\left[ \sum_{x\in\X}\frac{\pi_{x}^{\alpha}C_{x,y}^{\alpha}}{{\sum_{x\in\X}\pi_{x}^{\alpha}}} \right].\label{eq:a=balphabta}
    \end{align}
    To achieve the $\sup_{\pi}$, we apply the approach in~\cite[Thm. 3]{2022EXplainEps} as follows: Let $\displaystyle x^{*}=\argmax_{x}C_{x,y}$, and define a sequence of priors as:
    \begin{align}\label{eq:pin}
        \pi_{x^{*}}^{n}=1-\frac{1}{n}, \quad \pi_{x}^{n}=\frac{1}{n\left(|\X|-1\right)}~~\text{for}~~x\neq x^{*}.
    \end{align}
    It is evident that $\pi^n$ has full support, and also that $\displaystyle \lim_{n\rightarrow\infty}\sum_{x\in\X}\pi_{x}^{n}C_{x,y}=\max_{x\in\X}C_{x,y}$. Additionally, if we let $$\pi_{\alpha}=\frac{(\pi_{x})^{\alpha}}{\sum_{x\in\X}(\pi_{x})^{\alpha}},$$ then we have:
    \begin{align}\label{eq:pinalpah}
        \pi_{\alpha}^{n\rightarrow \infty}=\begin{cases}
            1, &x=x^{*},\\
            0, &x\neq x^{*}.
        \end{cases}
    \end{align}
     Applying $\pi^{n}$ in~\eqref{eq:a=balphabta}, by~\eqref{eq:pin} and~\eqref{eq:pinalpah} we get:
     \begin{align}
         \Lk^{\times}_{h_{(\alpha,\alpha)},f_{\alpha},g}(\forall,C)=\max_{x,x'}\frac{1}{\alpha-1}\log\sum_{y\in\Y}C_{x',y}^{1-\alpha}C_{x,y}^{\alpha}.
     \end{align}

\section{Proof of Theorem~\ref{thm:LDPf-1}}\label{app:proofofLDPf-1}
    Firstly, consider the following result for $\widehat{V}_{f,g}^{\max}[\pi,C]$ and increasing \(f\) and $f^{-1}$: 
    \begin{align}
        &\max_{y\in\Y} V_{f,g}(\delta^{y})=\max_{y\in\Y}\sup_{w\in\W}f^{-1}\bigg(\sum_{x\in\X}\delta_{x}^{y} f\big(g(w,x)\big)\bigg)\\
        &=  \max_{y\in\Y} \sup_{w\in\W}f^{-1}\bigg(\sum_{x\in\X}\frac{\pi_{x}C_{x,y}}{p(y)} f\big(g(w,x)\big)\bigg)\\
        &\leq  \max_{y\in\Y} \sup_{w\in\W} f^{-1}\bigg( \Big(\frac{\max_{x\in\X}C_{x,y}}{p(y)} \Big) \sum_{x\in\X} \pi_{x}f\big(g(w,x)\big)\bigg) \label{eq:maxcxyp(y)}\\
        &=\max_{y\in\Y}  f^{-1}\left( \frac{\max_{x\in\X}C_{x,y}}{p(y)}\right) \sup_{w\in\W} f^{-1}\bigg( \sum_{x\in\X} \pi_{x}f\big(g(w,x)\big)\bigg) \label{eq:multiplicativefh}\\
        &= \max_{y\in\Y}  f^{-1}\bigg( \frac{\max_{x\in\X}C_{x,y}}{p(y)}\bigg) V_{f,g}(\pi),
    \end{align}
    For the generalized max-case capacity $\mathcal{L}^{\max}_{f,\forall}(\forall,C)$, we have:
    \begin{align}
        &\mathcal{L}^{\max}_{f,\forall}(\forall,C)=\sup_{\pi,g}\mathcal{L}^{\max}_{f,g}(\pi,C)=\log \sup_{\pi,g} \frac{\widehat{V}_{f,g}^{\max}[\pi,C]}{V_{f,g}(\pi)}\\
        &\leq \log \sup_{\pi,g} \frac{\max_{y\in\Y}  f^{-1}\left( \frac{\max_{x\in\X}C_{x,y}}{p(y)}\right) V_{f,g}(\pi)}{V_{f,g}(\pi)}\\
        &=\log \max_{y\in\Y} f^{-1}\left( \sup_{\pi}\frac{\max_{x\in\X}C_{x,y}}{p(y)}\right)\\
        &=\log \max_{y\in\Y} f^{-1}\left( \frac{\max_{x\in\X}C_{x,y}}{\min_{x\in\X}C_{x,y}}\right)\label{eq:f-1LDP},
    \end{align}
    where~\eqref{eq:f-1LDP} has been proved in~\cite[Thm. 3]{2022EXplainEps}.
    If \(f\) and $f^{-1}$ are decreasing, we replace $\displaystyle \frac{\max_{x\in\X}C_{x,y}}{p(y)}$ with $\displaystyle \frac{\min_{x\in\X}C_{x,y}}{p(y)}$ in~\eqref{eq:maxcxyp(y)} and the result is achieved.

\section{Proof of Proposition~\ref{prop:AVG+CVX=DPI}}\label{app:proofofAVG+CVX=DPI}

        Assume \(\mathcal{X}\), \(\mathcal{Y}\), and \(\mathcal{Z}\) are sets of possible values. 
        Let \(\pi\) represent a prior distribution over \(\mathcal{X}\), \(C\) denote a channel from \(\mathcal{X}\) to \(\mathcal{Y}\), and \(R\) be a channel from \(\mathcal{Y}\) to \(\mathcal{Z}\). 
        The sequential combination of channels \(C\) and \(R\), symbolized by \(C R\), forms a new channel that maps \(\mathcal{X}\) to \(\mathcal{Z}\). 
        Consequently, the corresponding inner of $[\pi,C]$ for each $y \in \Y$ is  $\delta^{y}=p(X|y)$ and the corresponding inner of $[\pi,CR]$ for each $z \in \mathcal{Z}$ is $\delta^{z}=p(X|z).$ 
        Define the joint probability distribution \(p(x, y, z)\) as $p(x,y,z) = \pi_x C_{x,y} R_{y,z}$ for each \((x, y, z) \in \mathcal{X} \times \mathcal{Y} \times \mathcal{Z}\). 
        This joint distribution makes a Markov chain $X-Y-Z$. Thus, we have \(p(z|x, y) = p(z|y)\) and  $p(x|z)=\sum_{y\in\Y}p(x|y)p(y|z)$.
        
        First, assume $h\neq f$ and let $h$ be either convex and increasing or concave and decreasing. In the following, we consider the first case:
        \begin{align}
            & \widehat{V}_{h,f,g}[\pi, C] =
            h^{-1}\bigg( \sum_{y\in\Y} p(y) h\Big( V_{f,g}\big(p_{X|y}\big) \Big)\bigg) \nonumber\\
            &= h^{-1}\bigg( \sum_{y\in\Y} \Big(\sum_{z\in\mathcal{Z}}p(z)p(y|z)\Big) h\Big( V_{f,g}\big(p_{X|y}\big)\Big) \bigg) \nonumber\\
            &= h^{-1}\Bigg(\sum_{z\in\mathcal{Z}}p(z) \bigg( \sum_{y\in\Y} p(y|z) h\Big( V_{f,g}\big(p_{X|y}\big)\Big)\bigg)\Bigg) \\
            &\geq h^{-1}\bigg(\sum_{z\in\mathcal{Z}}p(z) h\Big(\sum_{y\in\Y} p(y|z)V_{f,g}\big(p_{X|y}\big)\Big)\bigg) \label{eq:ineq-h}\\
            &\geq h^{-1}\Bigg(\sum_{z\in\mathcal{Z}}p(z) h \bigg( V_{f,g}\Big(\sum_{y\in\Y} p(y|z)p(x|y)\Big) \bigg) \Bigg) \label{eq:ineq-Vfg}\\
            &= h^{-1}\bigg(\sum_{z\in\mathcal{Z}}p(z)  h\Big( V_{f,g}\big(p_{X|z}\big)\Big) \bigg)= \widehat{V}_{h,f,g}[\pi,CR],\nonumber
        \end{align}
        where~\eqref{eq:ineq-h} holds since $h$ is convex and increasing which implies that $h^{-1}$ is also  increasing. Similarly,~\eqref{eq:ineq-Vfg} is due to the convexity of $V_{f,g}(p_{X|y})$ and $h$ and $h^{-1}$ being increasing. When $h$ is concave and decreasing, the same inequalities hold.
        
        When $h=f$, it can be either convex and decreasing or concave and increasing, which are duals of the case when $h\neq f$. However, \DPI~still holds due to~\eqref{eq:h=f}. We prove \DPI~for increasing \(f\). The same proof applies to the decreasing case. 
        \begin{align*}
            &\widehat{V}_{f,f,g}[\pi, C]=f^{-1}\bigg(\sum_{y\in\Y} p(y) \Big( \sup_{w\in\W} \sum_{x\in\X} \delta_{x}^{y} f\big(g(w, x)\big) \Big)\bigg) \\
            &=f^{-1}\bigg(\sum_{z\in\mathcal{Z}}p(z)\sum_{y\in\Y}p(y|z) \Big( \sup_{w\in\W} \sum_{x\in\X} p(x|y) f\big( g(w,x)\big)\Big)\bigg) \\
            &\geq f^{-1}\bigg(\sum_{z\in\mathcal{Z}}p(z)\Big( \sup_{w\in\W} \sum_{x\in\X} \sum_{y\in\Y}p(y|z)p(x|y) f\big(g(w,x)\big)\Big)\bigg) \\
            &= f^{-1}\bigg(\sum_{z\in\mathcal{Z}}p(z)\Big( \sup_{w\in\W} \sum_{x\in\X} p(x|z) f\big(g(w,x)\big) \Big)\bigg) \\
            &=\widehat{V}_{f,f,g}[\pi, CR],\nonumber
        \end{align*}
        where the inequality is true since $\displaystyle \sup_{w\in\W}\sum_{x\in\X} p(x|y) f\big( g(w, x) \big)$ is convex according to the proof of Theorem~\ref{thm:Vfg axioms}.

\end{document}